\documentclass[lettersize,journal]{IEEEtran}
\usepackage{amsmath,amsfonts,amssymb}
\usepackage{algorithmic}
\usepackage{algorithm}
\usepackage{array}
\usepackage[caption=false,font=normalsize,labelfont=sf,textfont=sf]{subfig}
\usepackage{textcomp}
\usepackage{stfloats}
\usepackage{url}
\usepackage{verbatim}
\usepackage{graphicx}
\usepackage{cite}
\usepackage{booktabs}
\usepackage{multirow}
\usepackage{bm}
\usepackage{mathrsfs}
\usepackage{mathtools}
\usepackage[dvipsnames,table,xcdraw]{xcolor}
\usepackage{colortbl}
\usepackage{makecell}
\usepackage{pifont}
\usepackage{hyperref}
\usepackage{amsthm}

\newtheorem{theorem}{Theorem}
\newtheorem{lemma}[theorem]{Lemma}
\newtheorem{proposition}[theorem]{Proposition}

\newtheorem{assumption}{Assumption}
\theoremstyle{remark}

\colorlet{phaseI}{rgb:red!2,65;green!30,60;blue!20,125}
\colorlet{phaseII}{rgb:red!2,65;green!30,90;blue!20,125}
\colorlet{phaseIII}{rgb:red!60,100;green!20,90;blue!30,125}
\newcommand{\gray}[1]{\textcolor{gray}{#1}}

\newcommand{\method}{\textbf{\texttt{RELIEF}}}
\newcommand{\best}[1]{\textcolor{red}{\textbf{#1}}}
\newcommand{\second}[1]{\textcolor{blue}{\textbf{#1}}}
\newcommand{\up}[1]{{\color{RedOrange}$\uparrow$#1}}
\newcommand{\down}[1]{{\color{BlueGreen}$\downarrow$#1}}
\newcommand{\venue}[1]{{\scriptsize\color{gray}[#1]}}
\newcommand{\std}[1]{$_{\pm\text{\scriptsize #1}}$}
\definecolor{deepred}{rgb}{0.6,0,0}

\begin{document}

\title{RELIEF: Turning Missing Modalities into Training Acceleration for Federated Learning on Heterogeneous IoT Edge}

\author{Beining~Wu,~\IEEEmembership{Member,~IEEE},
        Zihao~Ding, and
        Jun~Huang,~\IEEEmembership{Senior Member,~IEEE}%
\IEEEcompsocitemizethanks{
\IEEEcompsocthanksitem Beining Wu, Zihao Ding, and Jun Huang are with the Department of Electrical Engineering and Computer Science, South Dakota State University, Brookings, SD 57007, USA.
E-mails: \{Wu.Beining, Zihao.Ding\}@jacks.sdstate.edu, Jun.Huang@sdstate.edu.
}
\thanks{This work was supported in part by the National Science Foundation under grant CNS-2348422.}}

\markboth{IEEE Internet of Things Journal}%
{Wu \MakeLowercase{\textit{et al.}}: RELIEF}

\maketitle

\begin{abstract}
Federated learning (FL) over heterogeneous IoT edge devices faces coupled system-modality-data heterogeneity: the lower-cost device carries both fewer sensors and less computational power, so the slowest device (straggler) produces the most incomplete gradient signals. Naively averaging their updates dilutes rare-modality information and wastes computation on absent-sensor parameters, whereas existing methods handle the triple heterogeneity (system, modality, data) in isolation and none addresses their coupling. To resolve this issue, we propose \method{}, a framework that partitions the fusion-layer Low-Rank Adaptation (LoRA) projection matrix into modality-aligned column blocks and uses this partition as a unified interface for aggregation, elastic training, and communication. Each block is aggregated only within the cohort of devices possessing that modality, which eliminates cross-modal gradient interference; the server then allocates personalized training budgets by prioritizing blocks with the highest cohort-internal divergence, so that resource-constrained devices train fewer but more impactful parameters. We prove that cohort-wise aggregation removes interference from the convergence bound and that the divergence-guided allocation achieves sublinear regret. Experiments on two IoT sensor datasets (PAMAP2, MHEALTH) under both full-parameter (CNN) and parameter-efficient (LoRA) training show that \method{} achieves up to 9.41$\times$ speedup and 37\% energy reduction over FedAvg with up to 15.3\,pp rare-modality F1 gains, and real-device validation on a two-Jetson AGX Orin testbed confirms these results.
\end{abstract}

\begin{IEEEkeywords}
Federated learning, multimodal learning, Internet of Things, edge computing, low-rank adaptation, device heterogeneity
\end{IEEEkeywords}

\section{Introduction}
\label{sec:introduction}

\IEEEPARstart{T}{he} rapid growth of Internet-of-Things (IoT) deployments for health monitoring has created distributed networks of heterogeneous edge devices, ranging from multi-sensor smartwatches to single-axis adhesive patches, that continuously collect physiological and environmental data~\cite{Nguyen2021COMST, Wu2026TNSE,  Wu2026MNET}. Training shared sensing models across such networks demands collaborative learning without transmitting raw health data, making federated learning (FL) a natural framework~\cite{McMahan2017AISTATS, Kairouz2021FTML, Wu2026COMST}. However, real-world IoT networks exhibit a characteristic that distinguishes them from standard FL settings: three dimensions of device heterogeneity, namely system (computational capacity), modality (sensor availability), and data (non-IID distributions), are not independent but \emph{coupled through the device cost gradient}. Lower-cost devices carry both fewer sensors and less computational power, so the slowest devices are precisely those with the fewest modalities and the most incomplete gradient signals. As illustrated in Fig.~\ref{fig:problem}, this coupling creates challenges that existing methods, each addressing at most one dimension in isolation~\cite{Zhang2025NeurIPS, Ouyang2023MobiSys, Guo2025ICLR}, cannot resolve.

\begin{figure}[t]
\centering
\includegraphics[width=0.9\columnwidth]{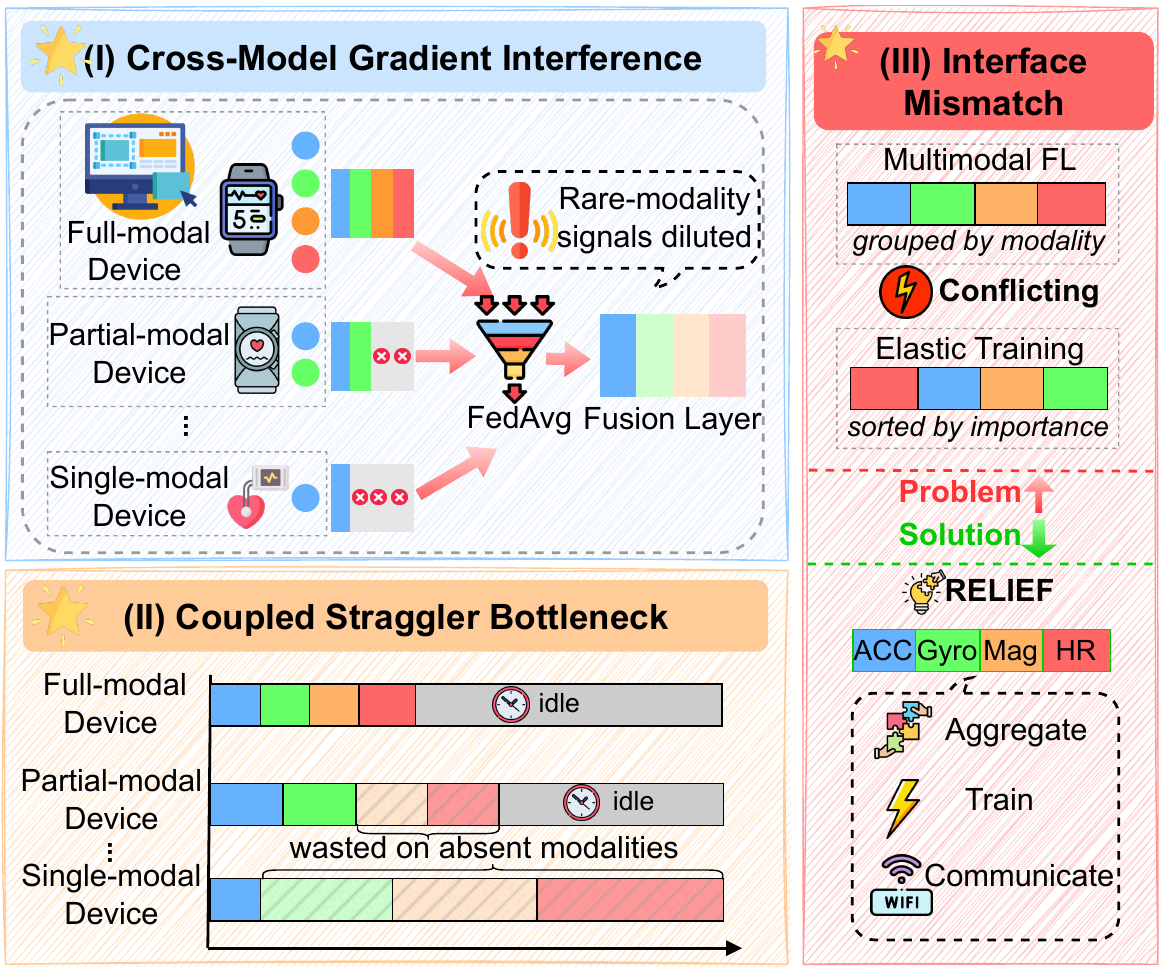}
\caption{Problem illustration. (Q1)~FedAvg dilutes rare-modality signals by mixing incompatible gradient updates. (Q2)~Single-modal stragglers waste computation on absent-modality parameters. (Q3)~Multimodal FL and elastic training define parameter groups along conflicting axes. \method{} uses modality-aligned column blocks as a unified interface for all three.}
\label{fig:problem}
\end{figure}
When devices with heterogeneous modality configurations jointly train a shared multimodal fusion model via standard FL aggregation~\cite{McMahan2017AISTATS}, the fusion layer receives structurally incompatible gradient updates: a full-modality device produces gradients encoding cross-modal interactions across all column blocks of the Low-Rank Adaptation (LoRA) projection matrix~\cite{Hu2022ICLR}, while a single-modality device produces non-zero gradients in only one block and near-zero noise elsewhere. Naively averaging these updates dilutes rare-modality signals and corrupts even shared-modality representations, as we empirically verify in Section~\ref{sec:motivation}. This raises the first question (Q1): \textit{How can we aggregate fusion-layer updates from devices with heterogeneous modality configurations without cross-modal interference?} Even if this aggregation problem were solved, synchronous FL still requires all devices to complete training before the next round can begin. Under the coupled heterogeneity described above, the bottleneck device, typically the one with the least compute and fewest sensors, must train the \emph{entire} model including parameter groups for modalities it does not possess, resulting in both wasted computation and prolonged round times. This leads to the second question (Q2): \textit{How can we accelerate training when the straggler devices are precisely those with the fewest modalities and least compute?} A natural approach is to combine existing solutions: use multimodal FL methods~\cite{Ouyang2023MobiSys, Yang2025ARXIV} to fix aggregation and elastic training methods~\cite{Zhang2025NeurIPS} to fix speed. However, these two lines of work define ``parameter groups'' differently: the former along modality boundaries, the latter along tensor importance, creating conflicting allocation decisions. An importance-based elastic scheduler, unaware of modality structure, may assign a single-modality device to train parameters for sensors it lacks, producing zero-gradient updates that waste compute and degrade aggregation quality. This gives rise to the fundamental question (Q3): \textit{Can a single structural interface inform what to aggregate, what to train, and what to communicate?}

Existing work falls short on all three fronts. In the multimodal FL literature, Harmony~\cite{Ouyang2023MobiSys} avoids federating the fusion layer altogether, sacrificing cross-modal knowledge sharing; FediLoRA~\cite{Yang2025ARXIV} and Pilot~\cite{Xiong2025AAAI} improve modality-heterogeneous aggregation but assume homogeneous device capabilities and provide no acceleration mechanism~\cite{Feng2023KDD, Wang2026TMC}. In the system-heterogeneous FL literature, FedEL~\cite{Zhang2025NeurIPS} achieves wall-clock speedup through tensor-level elastic training but operates on single-modality models and selects parameters by gradient magnitude without modality semantics~\cite{Qu2025KDD, Zhang2025TMC, Liu2025TMC}. In the federated LoRA literature, FedSA-LoRA~\cite{Guo2025ICLR} and FedLEASE~\cite{Wang2025NeurIPS} optimize aggregation under data heterogeneity, but do not exploit multimodal structure~\cite{Fan2025TOIT, Bian2025ICCV}. No existing method addresses the coupled system-modality-data heterogeneity that characterizes real-world IoT deployments, and no work leverages the LoRA projection matrix's modality-aligned column-block structure as a unified interface for aggregation, training allocation, and communication.

We propose \method{}, a \underline{\textbf{R}}esource-\underline{\textbf{E}}fficient mu\underline{\textbf{L}}t\underline{\textbf{I}}modal \underline{\textbf{E}}dge \underline{\textbf{F}}ederated learning framework, where a single structural decomposition resolves all three questions. To address Q1, we introduce Modality-Decomposed LoRA (MDLoRA) with cohort-wise aggregation: the fusion-layer LoRA projection matrix is partitioned into modality-aligned column blocks, and each block is aggregated exclusively within the subset of devices possessing that modality, which eliminates cross-modal gradient interference. To address Q2, we design divergence-guided modality-aware elastic training: the server computes per-block cohort-internal divergence and assigns each device a personalized training budget that prioritizes high-disagreement blocks while naturally excluding parameters for absent modalities, so that resource-constrained devices train fewer but more impactful parameters. To address Q3, both mechanisms share the same modality-aligned column-block structure as their common interface: the block boundaries define the aggregation cohorts, the elastic allocation units, and the communication granularity, which produces gains that exceed the sum of the individual components.

In summary, we make the following contributions.

\begin{itemize}
    \item We identify and formalize the coupled system-modality-data heterogeneity problem in multimodal IoT edge FL, and reveal through diagnostic experiments that cross-modal gradient interference propagates beyond missing-modality blocks to corrupt shared-modality representations, while rare-modality update divergence amplifies rather than converges over training.

    \item We propose \method{}, a unified framework that leverages the modality-aligned column-block structure of the fusion-layer LoRA matrix as a shared interface for cohort-wise aggregation, divergence-guided modality-aware elastic training, and on-demand communication.

    \item We provide theoretical analysis showing that cohort-wise aggregation eliminates cross-modal interference from the convergence bound and that divergence-guided allocation achieves near-optimal regret relative to an offline oracle.

    \item We conduct experiments on two real-world IoT sensor datasets (PAMAP2 and MHEALTH), demonstrating that \method{} achieves significant wall-clock speedup over baselines while maintaining or improving classification accuracy, with particular gains on rare-modality performance.
\end{itemize}

The remainder of this paper is organized as follows. Section~\ref{sec:related_work} reviews related work on system-heterogeneous FL, multimodal FL, and federated LoRA aggregation. Section~\ref{sec:problem_formulation} presents the system model and problem formulation. Section~\ref{sec:method} describes the proposed \method{} framework, preceded by motivational studies. Section~\ref{sec:theory} provides the theoretical analysis. Section~\ref{sec:experiments} reports simulation results. Section~\ref{sec:deployment} validates the framework on a real-device testbed, and Section~\ref{sec:con} concludes the paper.
\section{Related Work}
\label{sec:related_work}

\subsection{Federated Learning under System Heterogeneity}

Federated learning (FL) enables distributed edge devices to collaboratively train a shared model by exchanging model updates rather than raw data~\cite{McMahan2017AISTATS, Kairouz2021FTML, Wu2026COMST,Fang2026GLOBECOM}. A persistent challenge in practical deployments is system heterogeneity, where differences in computational capacity and communication bandwidth across devices create straggler bottlenecks during synchronous training. Classical approaches mitigate the effects of data heterogeneity through proximal regularization~\cite{Li2020MLSys} or variance-reduced gradient correction~\cite{Karimireddy2020ICML}, but do not address the underlying disparity in per-device training speed. To tackle this, recent works adopt partial or elastic training strategies that allow resource-constrained devices to train only a subset of the model~\cite{Zhang2025TMC, Liu2025TMC}. Zhang~\emph{et al.}~\cite{Zhang2025NeurIPS} introduce FedEL, which selects important tensors within a coordinated runtime budget through a sliding-window mechanism, achieving up to 3.87$\times$ wall-clock speedup on heterogeneous Jetson devices. Other efforts reduce inference-time or communication costs on weak devices through early-exit strategies~\cite{Qu2025KDD}, adaptive layer-wise compression~\cite{Nguyen2021COMST}, and acceleration via multithreaded federated training~\cite{Huang2025TMC, Wu2025WASA}.

However, the above methods~\cite{Zhang2025NeurIPS, Qu2025KDD, Zhang2025TMC, Liu2025TMC,Wu2023ACCESS,Pan2023SCIS} assume single-modality models, where parameter selection criteria such as gradient magnitude carry no modality semantics. When the model contains parallel encoder branches for multiple sensor modalities, importance-based allocation may assign weak devices to train parameters for sensors they do not possess, wasting computation on zero-gradient updates~\cite{Wu2025ToN, Ding2025IPCCC}.

\subsection{Multimodal Federated Learning}

Heterogeneous sensors in IoT edge networks have motivated multimodal FL, where devices with different sensor suites collaboratively train a shared fusion model~\cite{Feng2023KDD, Lin2023MIR, Anagnostopoulos2024Access, Wu2023ACCESS, Wu2026MNET, Wang2026TON,Fang2025TON, Pudasaini2026HPSR}. A central difficulty is modality heterogeneity: devices possessing different subsets of modalities produce structurally incompatible gradient updates for the shared fusion layer, complicating aggregation and degrading model quality. Ouyang~\emph{et al.}~\cite{Ouyang2023MobiSys} address this by disentangling training into modality-wise and fusion-wise stages, but avoid federating the fusion layer entirely to sidestep aggregation conflicts. Subsequent works explore alternative strategies, including correlation-adaptive multimodal split networks~\cite{Chen2022KDD}, Shapley-value-based modality scheduling~\cite{Bian2024ARXIV}, and prototype-based compensation for missing modalities in IoT online learning~\cite{Wang2026TMC}. In the LoRA-based multimodal FL space, Yang~\emph{et al.}~\cite{Yang2025ARXIV} propose dimension-wise aggregation to handle missing modalities, while Xiong~\emph{et al.}~\cite{Xiong2025AAAI} build a federated multimodal instruction tuning framework that requires each client to load all task-specific adapters.

These methods~\cite{Ouyang2023MobiSys, Yang2025ARXIV, Xiong2025AAAI, Bian2024ARXIV} focus on improving aggregation quality under modality heterogeneity but assume homogeneous device capabilities, providing no training acceleration mechanism. In real-world IoT deployments, device cost gradients couple modality availability with computational capacity, so pure aggregation improvements cannot resolve the straggler bottleneck.

\subsection{Low-Rank Adaptation in Federated Learning}

Low-Rank Adaptation (LoRA)~\cite{Hu2022ICLR} has become the dominant parameter-efficient fine-tuning approach in FL owing to its low communication and memory footprint~\cite{Ning2025TNNLS, Li2025TMC, DingICNC2025,Fang2025JSAC,Fang2025ARXIV}. A central research question is how to aggregate LoRA matrices across heterogeneous clients without introducing bias. Guo~\emph{et al.}~\cite{Guo2025ICLR} discover an asymmetry between the $A$ and $B$ matrices and propose sharing only $A$ for aggregation, while Bian~\emph{et al.}~\cite{Bian2025ICCV} reconstruct the ideal global update to correct aggregation residuals. To support devices with different resource budgets, heterogeneous-rank strategies allow clients to fine-tune at different LoRA ranks and reconcile them during aggregation~\cite{Fan2025TOIT, Tian2024NeurIPS}. Further advances address aggregation noise through tensor decomposition~\cite{Li2025ICCV}, residual correction~\cite{Yan2025ICLR}, and adaptive expert allocation that clusters clients by representation similarity~\cite{Wang2025NeurIPS}.

These works optimize LoRA aggregation along the data heterogeneity dimension but do not involve multimodal model structures. No existing method exploits the modality-aligned column-block structure of the LoRA projection matrix as a unified interface for cohort-wise aggregation, elastic training allocation, and on-demand communication, which is the approach that \method{} introduces.
\section{System Model and Problem Formulation}
\label{sec:problem_formulation}

\subsection{System Model}

We consider a multimodal \emph{federated learning} system for health monitoring in heterogeneous IoT edge networks. A set of $N$ edge devices $\{\mathcal{C}_n\}_{n=1}^N$ collaboratively train a shared multimodal model under the coordination of a central server, without transmitting raw data. Each device $\mathcal{C}_n$ is equipped with a modality subset $\mathcal{M}_n \subseteq \{1, \ldots, M\}$ and characterized by its computational capacity (in floating-point operations per second, or FLOPS) and communication bandwidth. For instance, in wearable activity monitoring, devices range from multi-sensor smartwatches ($|\mathcal{M}_n|{=}4$) to single-axis adhesive patches ($|\mathcal{M}_n|{=}1$).

In real-world IoT deployments, the three dimensions of device heterogeneity are \emph{coupled through the device cost gradient}: lower-cost devices carry both fewer sensors and less computational power, so the slowest devices (which determine the synchronous training bottleneck) are precisely those with the fewest modalities (which produce the most incomplete gradient signals). Each device further collects data from a local non-IID distribution $\mathcal{P}_n$. These three forms of heterogeneity, namely system, modality, and data, jointly create two coupled challenges: at the \emph{per-round} scale, synchronous aggregation forces all devices to wait for the slowest participant while the fusion layer receives structurally incompatible gradients; at the \emph{cross-round} scale, the accumulation of these inefficiencies degrades convergence and wastes both computation and communication.

\subsection{Problem Formulation}

The shared model consists of modality-specific encoders $\{E_m\}_{m=1}^M$, a fusion layer, and a task head $\mathcal{H}$. Each encoder $E_m$ maps raw input from modality $m$ to a feature $\mathbf{h}_m \in \mathbb{R}^{d_m}$. The fusion layer takes the concatenated vector $\mathbf{h} = [\mathbf{h}_1; \ldots; \mathbf{h}_M] \in \mathbb{R}^{D}$, where $D = \sum_{m=1}^M d_m$, and produces a fused representation for classification by $\mathcal{H}$.

We adopt \emph{Low-Rank Adaptation} (LoRA) for parameter-efficient federated training, keeping pretrained weights $W_0$ frozen and learning a low-rank residual $\Delta W = BA$, where $B \in \mathbb{R}^{d_o \times \rho}$ and $A \in \mathbb{R}^{\rho \times d_i}$ with $\rho \ll \min(d_i, d_o)$. LoRA is applied to the fusion layer, each encoder, and the task head. Since the fusion layer input $D$ is the ordered concatenation of per-modality dimensions $d_1, \ldots, d_M$, its projection matrix $A \in \mathbb{R}^{\rho \times D}$ can be partitioned into $M$ contiguous blocks:
\begin{equation}
A = [A_1 \mid A_2 \mid \cdots \mid A_M], \quad A_m \in \mathbb{R}^{\rho \times d_m},
\label{eq:mdlora_decomp}
\end{equation}
where each block $A_m$ exclusively processes the feature from modality $m$. For a device lacking modality $m$ ($m \notin \mathcal{M}_n$), the input $\mathbf{h}_m$ is zero and $A_m$ receives no gradient. As we show in Section~\ref{sec:method}, this decomposition provides a structural interface that enables modality-aware elastic training, cohort-wise aggregation, and on-demand communication.

We define the \emph{modality cohort} $\mathcal{C}_m = \{n : m \in \mathcal{M}_n\}$ as the subset of devices possessing modality $m$. The trainable parameters are organized into \emph{parameter groups}: $M$ fusion-layer column blocks $\{A_m\}_{m=1}^M$, the shared projection $B$, the LoRA parameters of each encoder $E_m$ (with $L_m$ adapted layers), and the task head $\mathcal{H}$ (with $L_\mathcal{H}$ layers), totaling $G = M + 1 + \sum_{m=1}^M L_m + L_\mathcal{H}$ groups (the $+1$ accounts for $B$). Each device $\mathcal{C}_n$ can only update the subset $\mathcal{G}_n \subseteq \{1, \ldots, G\}$ corresponding to its modalities $\mathcal{M}_n$.

Training proceeds over $R$ rounds. At round $r$, the server distributes $\Theta^r$ to selected devices $\hat{\mathcal{C}}^r \subseteq \{\mathcal{C}_1, \ldots, \mathcal{C}_N\}$. Each device $\mathcal{C}_n \in \hat{\mathcal{C}}^r$ performs $E$ epochs of local optimization on its dataset $\mathcal{D}_n \sim \mathcal{P}_n$:
\begin{equation}
\min_{\{\theta_j\}_{j \in \mathcal{G}_n}} \; \mathbb{E}_{(\mathbf{x}, y) \sim \mathcal{D}_n} \bigl[ \ell\bigl(f_\Theta(\mathbf{x}^{\mathcal{M}_n}), \, y\bigr) \bigr],
\label{eq:local_obj}
\end{equation}
where $f_\Theta$ is the full model, $\mathbf{x}^{\mathcal{M}_n}$ zero-pads missing modalities, and $\ell$ is the cross-entropy loss. Each device then uploads $\{\Delta\theta_{n,j}\}_{j \in \mathcal{S}_n}$, where $\mathcal{S}_n \subseteq \mathcal{G}_n$ is determined by the elastic training budget (Section~\ref{sec:method}), and the server aggregates them to obtain $\Theta^{r+1}$.

The \textbf{goal} is to maximize classification accuracy while minimizing wall-clock time-to-accuracy (TTA) under joint system, modality, and data heterogeneity. Achieving this is difficult because FedAvg on the full fusion matrix $A$ averages updates from devices with different modality configurations: a full-modality device produces gradients across all $M$ column blocks, while a single-modality device produces non-zero gradients in only one block. Averaging these structurally incompatible updates conflates informative cross-modal signals with zero-padded regions, introducing a bias that drives the aggregated update away from the global optimum. The problem is compounded by the straggler effect: under synchronous FL, the slowest device, typically the one with fewest modalities, must train the entire model including parameter groups for absent modalities, wasting computation and degrading gradient quality. These two effects motivate a unified solution that addresses \emph{what to train}, \emph{how to aggregate}, and \emph{what to communicate}, which we present in the next section.
\section{Proposed Method}
\label{sec:method}

\subsection{Motivational Studies}
\label{sec:motivation}

\begin{figure}[t]
\centering
\includegraphics[width=0.85\columnwidth]{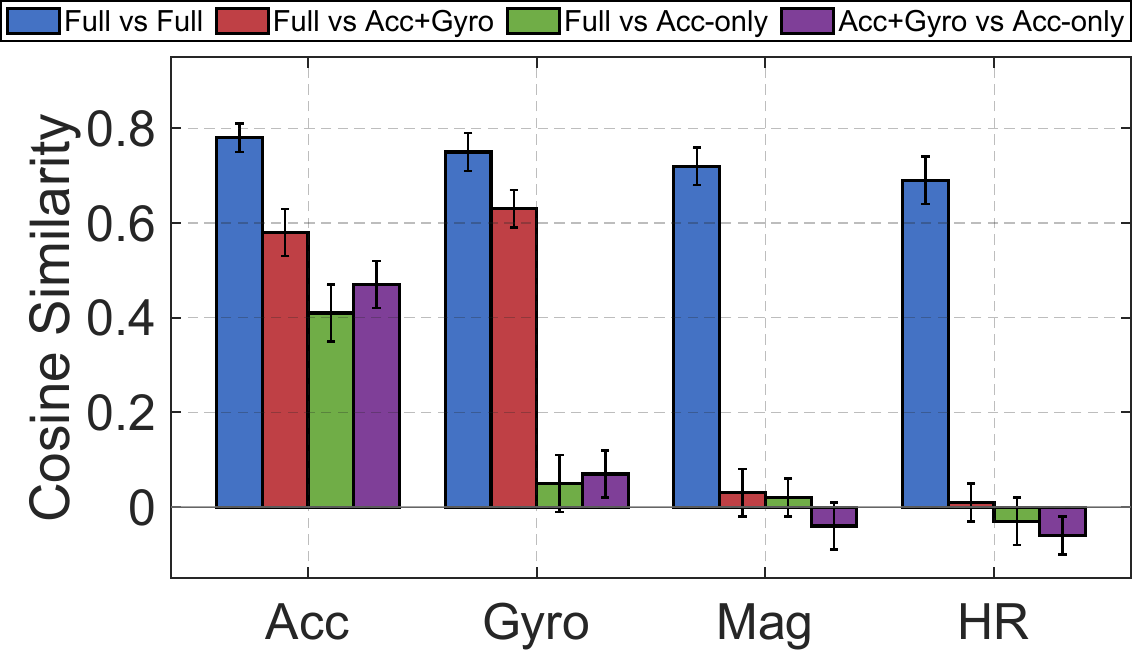}
\caption{Pairwise cosine similarity of fusion-layer LoRA updates across device pairs, grouped by modality column block. Interference extends beyond missing-modality blocks to shared blocks.}
\label{fig:motivation_cosine}
\end{figure}

To understand how modality heterogeneity affects federated training of multimodal fusion models, we conduct two diagnostic studies on the PAMAP2 activity recognition dataset~\cite{Reiss2012ISWC} under a heterogeneous IoT configuration: 8 clients partitioned into 3 device types (3$\times$Full with 4 modalities, 3$\times$Acc+Gyro, 2$\times$Acc-only), trained with standard FedAvg~\cite{McMahan2017AISTATS} for 200 rounds. The fusion layer adopts LoRA~\cite{Hu2022ICLR} with its projection matrix $A$ partitioned into four modality-aligned column blocks as defined in Eq.~\eqref{eq:mdlora_decomp}. These studies reveal two phenomena that are not addressed by existing methods.

\textbf{Observation 1: Gradient interference propagates from missing-modality to shared-modality blocks.}
We compute pairwise cosine similarity of $A$-matrix updates between device pairs, broken down by column block (Fig.~\ref{fig:motivation_cosine}). Blocks for absent modalities (Mag, HR) show near-zero similarity, as expected. Less expected is the Acc block, where all devices have the sensor: similarity between Full and Acc-only pairs drops to 0.41, far below the 0.78 between Full pairs. The full-modality gradient encodes cross-modal interactions absent from the unimodal Acc-only gradient; averaging them via FedAvg corrupts even shared-modality representations.

\textbf{Observation 2: Rare-modality divergence amplifies over training.}
We track cohort-internal update divergence of each column block across five training phases (Fig.~\ref{fig:motivation_divergence}). The Acc block (cohort size 8) maintains low, stable divergence throughout. In contrast, the Mag and HR blocks (cohort size 3) start high and \emph{continue to grow} with widening variance, because the small cohort amplifies aggregation noise from zero-padded gradients. Uniform parameter allocation that treats all blocks equally will under-serve rare modalities whose divergence demands more training attention.

\begin{figure}[t]
\centering
\includegraphics[width=0.85\columnwidth]{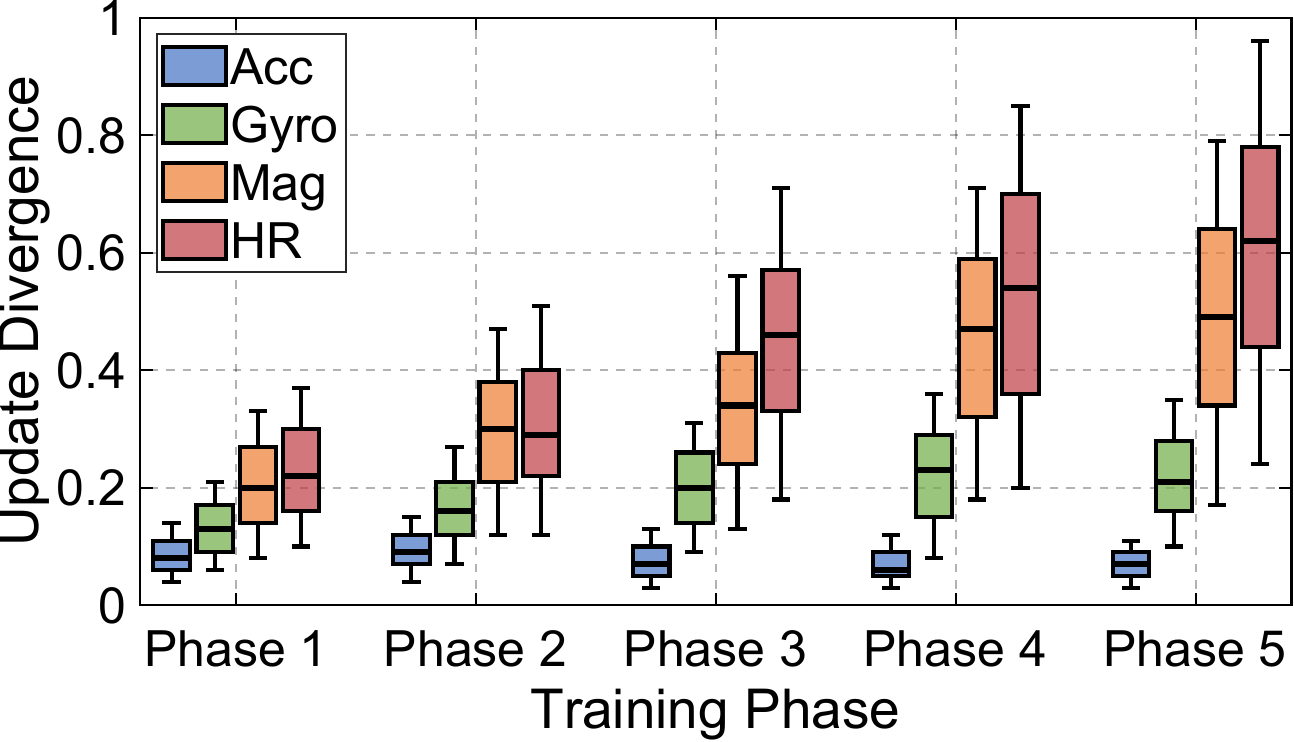}
\caption{Update divergence of each modality column block across training phases. Rare-modality blocks (Mag, HR) exhibit amplifying divergence rather than convergence.}
\label{fig:motivation_divergence}
\end{figure}

Together, these observations motivate a framework that addresses \emph{how to aggregate} (cohort-wise decomposition to eliminate the interference identified in Observation~1) and \emph{what to train} (divergence-guided modality-aware allocation to actively manage the amplifying divergence identified in Observation~2). Existing approaches either avoid federating the fusion layer altogether~\cite{Ouyang2023MobiSys} or apply modality-unaware parameter selection that cannot distinguish meaningful gradients from zero-padded noise. We present our unified solution in the following subsection.

\subsection{The \method{} Framework}
\label{sec:relief_framework}

Fig.~\ref{fig:framework} illustrates the \method{} architecture, which operates in a cyclic three-step protocol: (1)~the server computes per-group divergence within each modality cohort and assigns personalized elastic training budgets; (2)~each device trains only its assigned parameter groups; (3)~the server aggregates updates through modality-decomposed cohort-wise averaging. All three steps share one structural interface, the modality-aligned column blocks defined in Eq.~\eqref{eq:mdlora_decomp}, which serves as the aggregation boundary, the elastic allocation unit, and the communication granularity.

\begin{figure*}[t]
\centering
\includegraphics[width=0.9\textwidth]{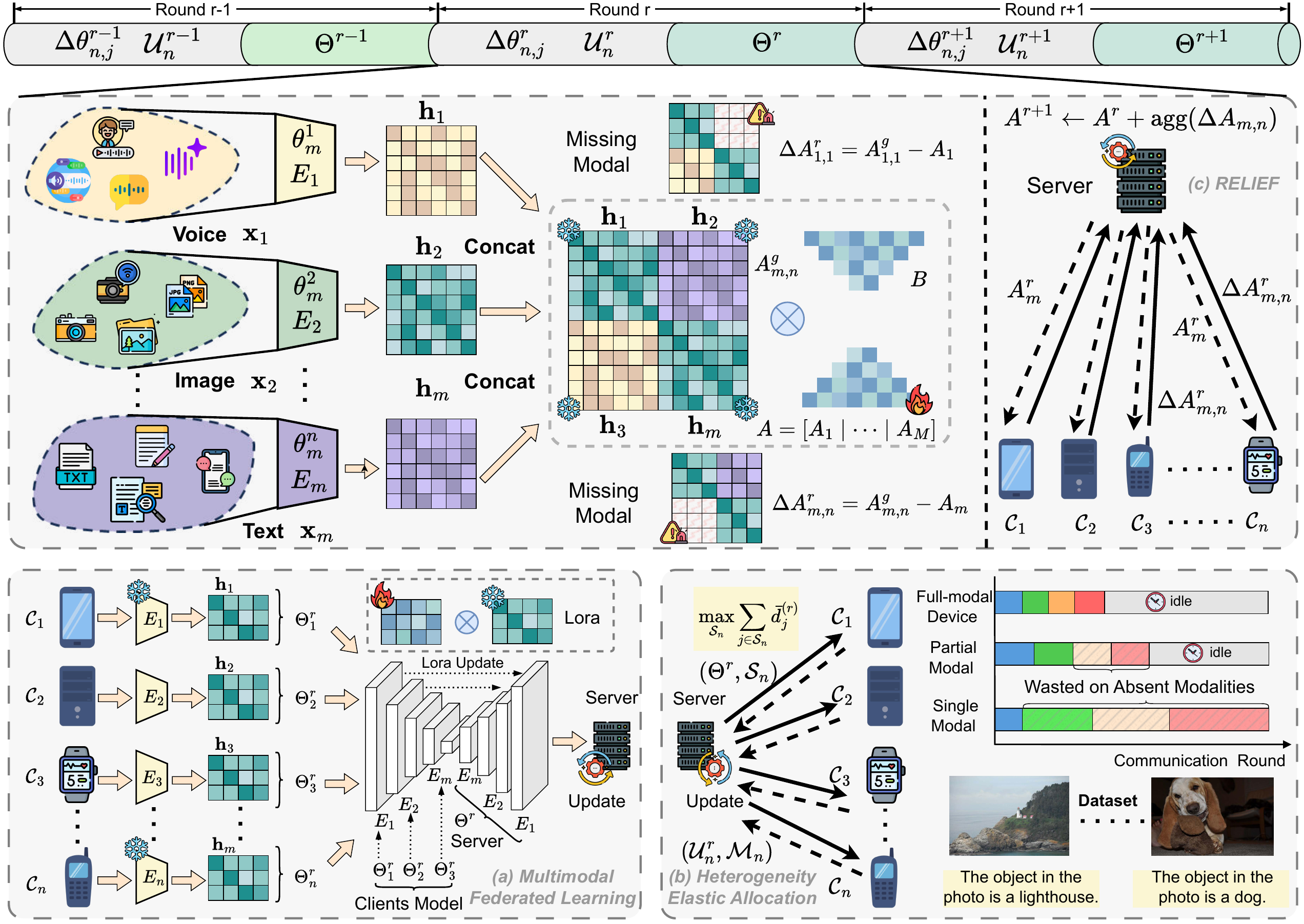}
\caption{Overview of \method{}. \textbf{Top:} Modality-aligned column-block decomposition of the fusion-layer LoRA matrix and cohort-scoped update aggregation across rounds. \textbf{Bottom:} (a) Multimodal FL with heterogeneous devices, (b) divergence-guided elastic allocation, and (c) cohort-wise server aggregation.}
\label{fig:framework}
\end{figure*}

\subsubsection{Modality-Decomposed LoRA (MDLoRA) with Cohort-Wise Aggregation}
\label{sec:mdlora}

Rather than averaging the full projection matrix $A$ across all devices as in FedAvg~\cite{McMahan2017AISTATS}, \method{} decomposes the aggregation along the modality-aligned column-block structure of $A$ (Eq.~\eqref{eq:mdlora_decomp}). Let $\tilde{\mathcal{C}}_m^r \subseteq \mathcal{C}_m$ denote the \emph{active cohort} at round $r$, i.e., the subset of devices that trained column block $A_m$ in this round. The aggregation for each block proceeds as:
\begin{equation}
A_m^{r+1} = A_m^{r} + \frac{1}{|\tilde{\mathcal{C}}_m^r|} \sum_{n \in \tilde{\mathcal{C}}_m^r} \Delta A_{m,n}^{r},
\label{eq:agg_A}
\end{equation}
where $\Delta A_{m,n}^r = A_{m,n}^r - A_m^r$ is the local update from device $n$. This eliminates the cross-modal interference identified in Observation~1: only devices possessing modality $m$ contribute to $A_m$, preventing zero-padded noise from corrupting meaningful gradient signals. Modality encoders $\{E_m\}$ follow the same cohort-wise rule.

The shared projection matrix $B \in \mathbb{R}^{d_o \times \rho}$ receives gradients from all participating devices. Since full-modality devices produce richer cross-modal projection signals, we aggregate $B$ with normalized modality-count weighting:
\begin{equation}
B^{r+1} = B^{r} + \sum_{n \in \hat{\mathcal{C}}^r} \underbrace{\frac{|\mathcal{M}_n| \,/\, M}{\textstyle\sum_{k \in \hat{\mathcal{C}}^r} |\mathcal{M}_k| \,/\, M}}_{\displaystyle w_n} \cdot \, \Delta B_n^r.
\label{eq:agg_B}
\end{equation}

This assigns higher weight to devices whose gradients span the full cross-modal projection space, since single-modality devices cannot capture these interactions. The task head $\mathcal{H}$ is aggregated via standard averaging across all devices.

\subsubsection{Divergence-Guided Modality-Aware Elastic Training}
\label{sec:elastic}

Cohort-wise aggregation resolves the interference problem but does not address the straggler bottleneck. Without elastic training, synchronous FL waits for the slowest device to finish training all its available parameter groups. \method{} addresses this through divergence-guided allocation that assigns each device a personalized subset of groups, prioritized by their cohort-internal disagreement.

\paragraph{Cohort-internal divergence.}
At each round $r$, the server quantifies how much devices within each cohort disagree on the update direction for each parameter group. For the fusion-layer block $A_m$, the cohort-internal divergence is defined as:
\begin{equation}
d_m^{r} = \frac{1}{|\mathcal{C}_m|} \sum_{n \in \mathcal{C}_m} \left\| \Delta A_{m,n}^{r-1} - \frac{1}{|\mathcal{C}_m|} \sum_{k \in \mathcal{C}_m} \Delta A_{m,k}^{r-1} \right\|_F^2.
\label{eq:div}
\end{equation}

Divergence for encoder layers $d_{m,l}^{r}$ and task head layers $d_{h,l}^{r}$ is computed analogously, with encoders using their modality cohort and the task head using all devices. To reduce sensitivity to per-round fluctuations, the server applies exponential moving average (EMA) smoothing:
\begin{equation}
\bar{d}_j^{r} = \gamma \cdot d_j^{r} + (1 - \gamma) \cdot \bar{d}_j^{r-1}, \quad \gamma \in (0, 1),
\label{eq:ema}
\end{equation}
where $j$ indexes over all parameter groups and $\gamma$ controls the balance between responsiveness and stability.

\paragraph{Personalized elastic allocation.}
Given the smoothed divergence estimates $\{\bar{d}_j^{r}\}$, the server generates a personalized training assignment $\mathcal{S}_n \subseteq \mathcal{G}_n$ for each device $\mathcal{C}_n$, where $\mathcal{G}_n$ is the set of parameter groups accessible to device $n$ (determined by $\mathcal{M}_n$). The assignment maximizes the total divergence $\sum_{j \in \mathcal{S}_n} \bar{d}_j^{r}$ subject to a budget constraint $|\mathcal{S}_n| \leq k_n$ and a mandatory inclusion constraint $\{A_m : m \in \mathcal{M}_n\} \subseteq \mathcal{S}_n$ that ensures every available fusion-layer block is trained. The mandatory set has cardinality $|\mathcal{M}_n|$, which is naturally smaller for devices with fewer sensors: a single-modality device must train at least one block, while a full-modality device must train at least $M$. Since this is a top-$k$ selection with mandatory inclusions, it admits a greedy solution: include the mandatory set, then fill the remaining $k_n - |\mathcal{M}_n|$ slots by descending $\bar{d}_j$.

The elastic budget $k_n$ is determined by the device's computational capacity relative to the per-round time target $T^*$:
\begin{equation}
k_n = \max\!\left(|\mathcal{M}_n|, \; \left\lfloor \frac{T^* - T_o}{\tau_n} \right\rfloor\right),
\label{eq:budget}
\end{equation}
where $T_o$ is the communication and synchronization overhead, and $\tau_n$ is the profiled per-group training time on device $n$. The value $T^*$ is selected via binary search to minimize the maximum per-round time across all devices.

\paragraph{Local training and communication.}
Each device $\mathcal{C}_n$ receives $(\Theta^r, \mathcal{S}_n)$ from the server and performs $E$ local epochs. Forward passes use the full model with zero-padded inputs for missing modalities, while gradient computation and parameter updates are restricted to $\mathcal{S}_n$. Upon completion, the device uploads only the trained groups:
\begin{equation}
\mathcal{U}_n^r = \left\{(j, \, \Delta\theta_{n,j}^r) : j \in \mathcal{S}_n \right\}, \quad |\mathcal{U}_n^r| \leq k_n,
\label{eq:upload}
\end{equation}
together with its modality configuration $\mathcal{M}_n$. A single-modality device uploads $|\mathcal{M}_n|/M$ of a full-modality device's volume, which reduces communication proportionally. Since resource-constrained devices also train fewer parameter groups per round, their average power draw decreases, and the energy savings compound with the wall-clock speedup.

\subsection{Training Pipeline}
\label{sec:pipeline}

\begin{algorithm}[h]
\caption{The \method{} Framework (\colorbox{blue!12}{Server Allocation}, \colorbox{green!12}{Local Training}, \colorbox{orange!12}{Cohort-Wise Aggregation})}
\label{alg:relief}
\textbf{Input:} $N$ devices $\{\mathcal{C}_n\}$ with modality sets $\{\mathcal{M}_n\}$, $R$ rounds, EMA coefficient $\gamma$, time target $T^*$
\begin{algorithmic}[1]
\STATE \textbf{Initialization:} All devices perform one full-training round; server computes initial divergence $\{\bar{d}_j^{0}\}$

\FOR{$r = 1$ to $R$}

    \colorbox{blue!12}{
    \parbox{0.85\columnwidth}{
    \STATE \gray{$\triangleright$ \textit{Server: divergence-guided elastic allocation}}
    \FOR{each parameter group $j$}
        \STATE Compute $d_j^{r}$ within its cohort (Eq.~\eqref{eq:div})
        \STATE Smooth: $\bar{d}_j^{r} \leftarrow \gamma \, d_j^{r} + (1-\gamma) \, \bar{d}_j^{r-1}$
    \ENDFOR
    \FOR{each device $\mathcal{C}_n \in \hat{\mathcal{C}}^r$}
        \STATE Compute $k_n$ via Eq.~\eqref{eq:budget}
        \STATE Solve allocation $\rightarrow$ $\mathcal{S}_n$ (top-$k_n$ by $\bar{d}_j$)
        \STATE Send $(\Theta^r, \mathcal{S}_n)$ to $\mathcal{C}_n$
    \ENDFOR
    }}

    \colorbox{green!12}{
    \parbox{0.85\columnwidth}{
    \STATE \gray{$\triangleright$ \textit{Devices: modality-aware local training}}
    \FOR{each $\mathcal{C}_n \in \hat{\mathcal{C}}^r$ \textbf{in parallel}}
        \FOR{epoch $= 1$ to $E$}
            \FOR{batch $(\mathbf{x}, y) \sim \mathcal{D}_n$}
                \STATE Forward with zero-padded missing modalities
                \STATE Backward and update only $\{j \in \mathcal{S}_n\}$
            \ENDFOR
        \ENDFOR
        \STATE Upload $\mathcal{U}_n^r$ and $\mathcal{M}_n$ to server
    \ENDFOR
    }}

    \colorbox{orange!12}{
    \parbox{0.85\columnwidth}{
    \STATE \gray{$\triangleright$ \textit{Server: cohort-wise aggregation}}
    \FOR{$m = 1, \ldots, M$}
        \STATE $A_m^{r+1} \leftarrow A_m^r + \frac{1}{|\tilde{\mathcal{C}}_m^r|} \sum_{n \in \tilde{\mathcal{C}}_m^r} \Delta A_{m,n}^r$
        \STATE Aggregate $E_m$ within $\tilde{\mathcal{C}}_m^r$ analogously
    \ENDFOR
    \STATE Aggregate $B^{r+1}$ with weights $\{w_n\}$ (Eq.~\eqref{eq:agg_B})
    \STATE Aggregate $\mathcal{H}^{r+1}$ via standard averaging
    }}

\ENDFOR
\RETURN $\Theta^R$
\end{algorithmic}
\end{algorithm}

The complete procedure is presented in Algorithm~\ref{alg:relief}. After a one-round initialization where all devices perform full training to bootstrap divergence estimates (line~1), the framework proceeds through $R$ rounds, each with three color-coded stages.

In the allocation stage (blue, lines~3--12), the server computes EMA-smoothed divergence for every parameter group within its modality cohort, then solves the personalized allocation for each selected device.

In the local training stage (green, lines~13--22), devices train in parallel on their assigned groups $\mathcal{S}_n$. Forward passes use the full model, but gradient updates are restricted to $\mathcal{S}_n$. Each device uploads only its trained groups (Eq.~\eqref{eq:upload}).

In the aggregation stage (orange, lines~23--29), each fusion-layer block $A_m$ and encoder $E_m$ is aggregated within its active cohort $\tilde{\mathcal{C}}_m^r$ via Eq.~\eqref{eq:agg_A}. The shared projection $B$ is aggregated with modality-count weighting, and the task head follows standard averaging. The server then proceeds to the next round with updated divergence estimates.
\section{Theoretical Analysis}
\label{sec:theory}

This section provides convergence and optimality guarantees for \method{}. Lemma~\ref{lem:decomp} decomposes FedAvg's aggregation error; Theorem~\ref{thm:cohort} shows that cohort-wise aggregation eliminates cross-modal interference; Theorem~\ref{thm:convergence} gives the convergence rate; and Propositions~\ref{prop:alloc}--\ref{prop:regret} establish the optimality and regret of the elastic allocation.

\subsection{Assumptions}

\begin{assumption}[Per-Group Smoothness]\label{asm:smooth}
For each parameter group $j \in \{1, \ldots, G\}$, the global loss $F$ is $L_j$-smooth with respect to $\theta_j$: for any $\theta_j, \theta_j'$,
\begin{equation}
\|\nabla_{\theta_j} F(\Theta) - \nabla_{\theta_j} F(\Theta')\| \leq L_j \|\theta_j - \theta_j'\|.
\end{equation}
We write $L = \max_j L_j$.
\end{assumption}

\begin{assumption}[Bounded Stochastic Variance]\label{asm:variance}
For each device $n$ and parameter group $j$, the stochastic gradient has bounded variance:
\begin{equation}
\mathbb{E}\!\left[\|\nabla_{\theta_j} f_n(\Theta; \mathbf{x}, y) - \nabla_{\theta_j} F_n(\Theta)\|^2\right] \leq \sigma^2.
\end{equation}
\end{assumption}

\begin{assumption}[Bounded Heterogeneity]\label{asm:hetero}
The local objectives deviate from the global objective by at most $\zeta^2$:
\begin{equation}
\frac{1}{N}\sum_{n=1}^N \|\nabla F_n(\Theta) - \nabla F(\Theta)\|^2 \leq \zeta^2.
\end{equation}
\end{assumption}

\begin{assumption}[Modality-Induced Gradient Structure]\label{asm:modality}
For a device $n$ lacking modality $m$ ($m \notin \mathcal{M}_n$), its update to fusion-layer block $A_m$ satisfies $\|\Delta A_{m,n}^r\|_F \leq \varepsilon_0$, where $\varepsilon_0 \to 0$ is the numerical noise from zero-padded inputs.
\end{assumption}

\noindent Assumptions~\ref{asm:smooth}--\ref{asm:hetero} are standard in federated optimization~\cite{McMahan2017AISTATS, Li2020MLSys}. Assumption~\ref{asm:modality} captures the near-zero gradients from missing sensors, as verified in Fig.~\ref{fig:motivation_cosine}.

\subsection{Aggregation Error Analysis}

\begin{lemma}[FedAvg Aggregation Error Decomposition]\label{lem:decomp}
Consider the fusion-layer block $A_m$ aggregated via FedAvg: $\hat{g}_m = \frac{1}{N}\sum_{n=1}^N \Delta A_{m,n}^r$. The expected squared error relative to the true cohort-mean gradient $\bar{g}_m = \frac{1}{|\mathcal{C}_m|}\sum_{n \in \mathcal{C}_m} \nabla_{A_m} F_n$ decomposes as:
\begin{equation}
\begin{aligned}
&\mathbb{E}\!\left[\|\hat{g}_m - \bar{g}_m\|_F^2\right] \\
&= \underbrace{\left\|\frac{|\mathcal{C}_m|}{N}\bar{g}_m + \frac{N{-}|\mathcal{C}_m|}{N}\mathbb{E}[\hat{\epsilon}_m] - \bar{g}_m\right\|_F^2}_{\mathrm{bias}^2} \\
&\quad + \underbrace{\frac{|\mathcal{C}_m|}{N^2}\cdot\frac{\sigma^2}{E} + \frac{N{-}|\mathcal{C}_m|}{N^2}\cdot\varepsilon_0^2}_{\mathrm{variance}},
\end{aligned}
\label{eq:decomp}
\end{equation}
where $\hat{\epsilon}_m = \frac{1}{N-|\mathcal{C}_m|}\sum_{n \notin \mathcal{C}_m}\Delta A_{m,n}^r$. The bias$^2$ term separates into scaling bias and cross-modal interference:
\begin{equation}
\begin{aligned}
\mathrm{bias}^2 &= \left(\frac{N{-}|\mathcal{C}_m|}{N}\right)^{\!2}\left\|\bar{g}_m - \mathbb{E}[\hat{\epsilon}_m]\right\|_F^2 \\
&\leq \underbrace{\left(1 - \frac{|\mathcal{C}_m|}{N}\right)^{\!2}\|\bar{g}_m\|_F^2}_{\mathrm{(I)\; scaling}} + \underbrace{\left(\frac{N{-}|\mathcal{C}_m|}{N}\right)^{\!2}\varepsilon_0^2}_{\mathrm{(II)\; interference}},
\end{aligned}
\label{eq:bias_expand}
\end{equation}
yielding the three-term decomposition: $\mathbb{E}[\|\hat{g}_m - \bar{g}_m\|_F^2] \leq \mathrm{(I)} + \mathrm{(II)} + \mathrm{(III)}$, where $\mathrm{(III)} = \frac{|\mathcal{C}_m|}{N^2}(\sigma^2/E + \zeta_m^2)$ with $\zeta_m^2 = \frac{1}{|\mathcal{C}_m|}\sum_{n \in \mathcal{C}_m}\|\nabla_{A_m} F_n - \bar{g}_m\|_F^2$.
\end{lemma}

\begin{proof}
Partition the $N$ devices into $\mathcal{C}_m$ and $\bar{\mathcal{C}}_m$. The FedAvg estimate separates as $\hat{g}_m = \frac{|\mathcal{C}_m|}{N}\hat{g}_m^{\mathcal{C}} + \frac{N-|\mathcal{C}_m|}{N}\hat{\epsilon}_m$. The standard bias-variance identity $\mathbb{E}[\|X - \mu\|^2] = \|\mathbb{E}[X] - \mu\|^2 + \mathrm{Var}(X)$ with $\mathbb{E}[\hat{g}_m^{\mathcal{C}}] = \bar{g}_m$, $\|\mathbb{E}[\hat{\epsilon}_m]\|_F \leq \varepsilon_0$ (Assumption~\ref{asm:modality}), and the relaxation $\|a - b\|^2 \leq 2\|a\|^2 + 2\|b\|^2$ yields~\eqref{eq:decomp}--\eqref{eq:bias_expand}. The variance uses Assumptions~\ref{asm:variance}--\ref{asm:modality} averaged over $E$ steps and the respective cohort sizes.
\end{proof}

Term~(I) dilutes the update by $|\mathcal{C}_m|/N$; Term~(II) quantifies cross-modal interference from zero-padded gradients (Observation~1); Term~(III) is the irreducible intra-cohort disagreement.

\begin{theorem}[Cohort-Wise Aggregation Error]\label{thm:cohort}
Under \method{}'s aggregation rule (Eq.~\eqref{eq:agg_A}), the expected squared error for block $A_m$ satisfies:
\begin{equation}
\mathbb{E}\!\left[\|\tilde{g}_m - \bar{g}_m\|_F^2\right]
\leq \frac{1}{|\tilde{\mathcal{C}}_m^r|}\left(\frac{\sigma^2}{E} + \zeta_m^2\right).
\label{eq:cohort_error}
\end{equation}
Terms~\emph{(I)} and~\emph{(II)} of Lemma~\ref{lem:decomp} vanish identically.
\end{theorem}

\begin{proof}
Since $\tilde{\mathcal{C}}_m^r \subseteq \mathcal{C}_m$, no zero-padded gradient enters the sum (Term~(II) = 0), and the $1/|\tilde{\mathcal{C}}_m^r|$ weight is the exact cohort average (Term~(I) = 0). The remaining error $\tilde{g}_m - \bar{g}_m = \frac{1}{|\tilde{\mathcal{C}}_m^r|}\sum_{n \in \tilde{\mathcal{C}}_m^r}(\Delta A_{m,n}^r - \bar{g}_m)$ separates into stochastic noise ($\leq \sigma^2/E$) and heterogeneity ($\zeta_m^2$) per device, with the $1/|\tilde{\mathcal{C}}_m^r|$ prefactor yielding~\eqref{eq:cohort_error}.
\end{proof}

\subsection{Convergence Analysis}

\begin{theorem}[Convergence of \method{}]\label{thm:convergence}
Let Assumptions~\ref{asm:smooth}--\ref{asm:modality} hold. With learning rate $\eta = O(1/\sqrt{ER})$ and $E$ local epochs per round, after $R$ rounds \method{} satisfies:
\begin{equation}
\begin{aligned}
&\frac{1}{R}\sum_{r=1}^{R} \mathbb{E}\!\left[\|\nabla F(\Theta^r)\|^2\right] \\
&\leq \underbrace{\frac{2\bigl(F(\Theta^0) - F^*\bigr)}{\eta E R}}_{\mathrm{optimization}} + \underbrace{4L\eta E\,\zeta^2}_{\mathrm{client\; drift}} + \underbrace{\frac{2L\eta\sigma^2}{N}}_{\mathrm{noise}} \\
&\quad + \underbrace{2L\eta\sum_{m=1}^{M}\frac{\sigma^2/E + \zeta_m^2}{\min_r|\tilde{\mathcal{C}}_m^r|}}_{\mathrm{cohort\; residual}}.
\end{aligned}
\label{eq:convergence}
\end{equation}
\end{theorem}

\begin{proof}
Starting from the $L$-smooth descent lemma, the global update decomposes across parameter groups (fusion blocks, encoders, task head), each aggregated within its respective cohort. Substituting Theorem~\ref{thm:cohort} for the fusion-layer blocks and standard FedAvg bounds~\cite{Li2020MLSys, Karimireddy2020ICML} for the remaining groups yields the per-round descent:
\begin{equation}
\begin{aligned}
&\mathbb{E}\!\left[F(\Theta^{r+1})\right] \leq \mathbb{E}\!\left[F(\Theta^r)\right] - \frac{\eta E}{2}\mathbb{E}\!\left[\|\nabla F(\Theta^r)\|^2\right] \\
&\quad + \frac{L\eta^2 E^2}{2}\zeta^2 + \frac{L\eta^2\sigma^2}{2N} + \frac{L\eta^2}{2}\sum_{m=1}^{M} \frac{\sigma^2/E + \zeta_m^2}{|\tilde{\mathcal{C}}_m^r|}.
\end{aligned}
\end{equation}
Rearranging and telescoping over $r = 0, \ldots, R{-}1$, then dividing by $\eta ER/2$, yields~\eqref{eq:convergence}.
\end{proof}

The cohort residual (last term) is the only aggregation-dependent term. FedAvg incurs additional cross-modal interference that does not vanish with more rounds. \method{} further reduces this term by enlarging $\min_r|\tilde{\mathcal{C}}_m^r|$ for high-divergence modalities via elastic allocation.

\subsection{Elastic Allocation Analysis}

\begin{proposition}[Optimality of Divergence-Guided Allocation]\label{prop:alloc}
Define the weighted cohort residual $\mathcal{R}(\{x_m\}) = \sum_{m=1}^M \Delta_m / x_m$, where $\Delta_m = \sigma^2/E + \zeta_m^2$ is the per-block divergence and $x_m = |\tilde{\mathcal{C}}_m|$. Under total budget $\sum_m x_m \leq K$ (the aggregate elastic budget across all devices), the optimal allocation and minimum residual are:
\begin{equation}
x_m^* = \frac{\sqrt{\Delta_m}}{\sum_{m'=1}^M \sqrt{\Delta_{m'}}} \cdot K, \qquad \mathcal{R}^* = \frac{\left(\sum_{m=1}^M \sqrt{\Delta_m}\right)^2}{K}.
\label{eq:opt_alloc}
\end{equation}
\end{proposition}

\begin{proof}
The Lagrangian $\mathcal{L} = \sum_m \Delta_m/x_m + \lambda(\sum_m x_m - K)$ has KKT stationarity condition $-\Delta_m/x_m^2 + \lambda = 0$, which gives $x_m^* = \sqrt{\Delta_m/\lambda}$. Substituting into the budget constraint:
\begin{equation}
\sum_{m=1}^M \sqrt{\frac{\Delta_m}{\lambda}} = K \;\;\implies\;\; \lambda^* = \frac{\left(\sum_{m=1}^M \sqrt{\Delta_m}\right)^2}{K^2}.
\end{equation}
Back-substituting yields $x_m^*$ and the optimal objective:
\begin{equation}
\mathcal{R}^* = \sum_{m=1}^M \frac{\Delta_m}{x_m^*} = \sum_{m=1}^M \frac{\Delta_m \sum_{m'}\sqrt{\Delta_{m'}}}{K\sqrt{\Delta_m}} = \frac{\left(\sum_{m=1}^M \sqrt{\Delta_m}\right)^2}{K}.
\end{equation}
Since $x_m^* \propto \sqrt{\Delta_m}$, the divergence-guided greedy allocation (selecting groups by descending $\bar{d}_j$) is a rank-preserving discrete approximation.
\end{proof}

\begin{proposition}[Regret of EMA-Based Divergence Tracking]\label{prop:regret}
Let $d_j^{r}$ denote the true divergence at round $r$, with temporal variation bounded by $\delta = \max_{j,r}|d_j^{r+1} - d_j^{r}|$. The EMA estimate $\bar{d}_j^{r} = \gamma \, d_j^{r} + (1{-}\gamma)\,\bar{d}_j^{r-1}$ induces cumulative regret:
\begin{equation}
\sum_{r=1}^R \bigl[\mathcal{R}(\mathcal{S}^r) - \mathcal{R}(\mathcal{S}^{*r})\bigr] \leq \frac{\gamma\,\delta\sqrt{R}}{(1 - \gamma)^2}.
\label{eq:regret}
\end{equation}
\end{proposition}

\begin{proof}
Unrolling the EMA recursion gives $\bar{d}_j^{r} = \gamma\sum_{s=0}^{r-1}(1{-}\gamma)^s d_j^{r-s} + (1{-}\gamma)^r\bar{d}_j^{0}$. The estimation bias is bounded by:
\begin{equation}
\begin{aligned}
|\bar{d}_j^{r} - d_j^{r}|
&\leq \gamma\sum_{s=1}^{\infty}s\,(1{-}\gamma)^s\,\delta = \frac{\gamma\,\delta}{(1{-}\gamma)^2},
\end{aligned}
\end{equation}
using $|d_j^{r-s} - d_j^{r}| \leq s\delta$. Since $\mathcal{R}$ is Lipschitz in the divergence inputs, follow-the-leader analysis~\cite{Shalev2012FTML} yields $O(\sqrt{R})$ cumulative regret with the stated coefficient. The $O(\sqrt{R})$ rate is sublinear, unlike the $O(R)$ regret of uniform or random allocation.
\end{proof}

\section{Performance Evaluation}
\label{sec:experiments}

\subsection{Experimental Setup}

\subsubsection{Datasets}
We evaluate \method{} on two publicly available multimodal human activity recognition (HAR) datasets that reflect realistic IoT sensor heterogeneity.

\textbf{PAMAP2}~\cite{Reiss2012ISWC} contains data from 9 subjects wearing inertial measurement unit (IMU) sensors, covering 12 activity classes. Each subject provides four sensor modalities: accelerometer (Acc), gyroscope (Gyro), magnetometer (Mag), and heart rate (HR). Following prior work~\cite{Reiss2012ISWC}, we exclude Subject~9 due to insufficient recording length and partition the remaining 8 subjects into 8 FL clients, configured as 3$\times$Full (4 modalities, 1$\times$ compute), 3$\times$Acc+Gyro (2 modalities, 13$\times$ slower), and 2$\times$Acc-only (1 modality, 55$\times$ slower), which reflects the coupled cost gradient described in Section~\ref{sec:problem_formulation}.

\textbf{MHEALTH}~\cite{Banos2014IWAAL} records data from 10 subjects with four modalities: Acc, Gyro, Mag, and electrocardiogram (ECG), across 12 activity classes. We partition by subject into 10 clients: 3$\times$Full (1$\times$), 3$\times$Acc+Gyro (13$\times$), and 4$\times$Acc-only (55$\times$).

Both datasets use a sliding window of 5.12\,s (256 samples at 50\,Hz) with a 1\,s stride. We report macro-F1 as the primary metric, wall-clock speedup relative to FedAvg, and per-round communication volume.

\subsubsection{Baselines}
We compare against 10 methods spanning three categories.

\textit{Classical FL:}
FedAvg~\cite{McMahan2017AISTATS} and
FedProx~\cite{Li2020MLSys}.

\textit{System-heterogeneous / elastic FL:}
FedEL~\cite{Zhang2025NeurIPS},
FedICU~\cite{Liao2025ICML}, and
DarkDistill~\cite{Qu2025KDD}.

\textit{Multimodal FL / federated LoRA:}
Harmony~\cite{Ouyang2023MobiSys},
Pilot~\cite{Xiong2025AAAI},
FedSA-LoRA~\cite{Guo2025ICLR},
HeLoRA~\cite{Fan2025TOIT}, and
FedLEASE~\cite{Wang2025NeurIPS}.

All methods share identical data splits, device heterogeneity configurations, and communication protocols.

\subsubsection{Implementation Details}
We employ two backbone architectures to validate \method{} under both full-parameter and parameter-efficient training.

\textbf{Backbone~1 (lightweight CNN):} Each modality encoder is a 2-layer 1D convolutional neural network (CNN) with $<$2M total parameters. The fusion layer is a fully connected layer whose weight matrix is partitioned into modality-aligned column blocks for modality-decomposed aggregation.

\textbf{Backbone~2 (pretrained Transformer + LoRA):} Each modality encoder is an independent MOMENT~\cite{Goswami2024ICML} instance ($\sim$40M parameters, $\sim$160M total). We freeze the backbone and inject LoRA adapters ($\rho = 8$) into each attention layer's Q/V projections and the feed-forward network (FFN), which yields $\sim$300K trainable parameters ($<$0.2\%). The fusion-layer LoRA projection matrix $A$ is partitioned into modality-aligned column blocks for MDLoRA.

Training uses Adam with learning rate $1\times10^{-3}$, batch size 32, $E = 5$ local epochs, and $R = 200$ rounds. The EMA coefficient is $\gamma = 0.9$. Device heterogeneity is simulated via FLOP-proportional profiling calibrated to edge tera operations per second (TOPS): Full devices at 275 TOPS (Jetson AGX Orin level), Acc+Gyro at 21 TOPS (Xavier NX level), and Acc-only at 5 TOPS (low-end IoT level).

We estimate per-round fleet energy as the sum of each device's active training power, communication power, and idle waiting power multiplied by their respective durations, with active power calibrated from Jetson AGX Orin datasheets (60\,W at MAXN, 30\,W at 30W mode, 15\,W at 15W mode, 5\,W for low-end IoT) and idle power set to 20\% of the active level. This datasheet-based model provides an approximate estimate; real-device energy profiling with hardware power monitors is presented in Section~\ref{sec:deployment}.

\subsection{Main Results}

\begin{table*}[t]
\centering
\setlength{\tabcolsep}{6pt}
\renewcommand\arraystretch{1.15}
\caption{Main results with Backbone~1 (lightweight CNN, full-parameter training). Best in \best{red}, runner-up in \second{blue}. $^\dagger$Rare-Mod F1: PAMAP2 avg(Mag, HR). TTA: rounds to reach 85\% F1.}
\label{tab:main_cnn}
\resizebox{0.95\linewidth}{!}{
\begin{tabular}{l|ccc|c|c|c|c}
\Xhline{1.2pt}
\rowcolor{CadetBlue!20}
\textbf{Method} & \textbf{PAMAP2 F1 (\%)} & \textbf{MHEALTH F1 (\%)} & \textbf{Rare-Mod F1 (\%)$^\dagger$} & \textbf{Speedup} & \textbf{TTA (rds)} & \textbf{Comm (MB)} & \textbf{Energy (J)} \\
\Xhline{1.2pt}
FedAvg \venue{AISTATS'17}~\cite{McMahan2017AISTATS} & \second{92.0}\std{0.38} & 91.6\std{0.45} & 37.5\std{0.92} & 1.00$\times$ & 75 & 4.81 & 847 \\
\rowcolor{gray!10}
FedProx \venue{MLSys'20}~\cite{Li2020MLSys} & \best{92.5}\std{0.31} (\up{0.5}) & 91.7\std{0.42} (\up{0.1}) & 38.2\std{0.88} (\up{0.7}) & 1.00$\times$ & 65 & 4.81 & 852 \\
\hline
FedEL \venue{NeurIPS'25}~\cite{Zhang2025NeurIPS} & 81.6\std{1.05} (\down{10.4}) & 62.2\std{1.42} (\down{29.4}) & 15.3\std{1.37} (\down{22.2}) & \best{6.83$\times$} & --- & \best{3.51} & \best{198} \\
\rowcolor{gray!10}
FedICU \venue{ICML'25}~\cite{Liao2025ICML} & 87.6\std{0.62} (\down{4.4}) & 89.6\std{0.53} (\down{2.0}) & 28.4\std{1.05} (\down{9.1}) & 1.15$\times$ & 120 & \second{3.85} & 761 \\
DarkDistill \venue{KDD'25}~\cite{Qu2025KDD} & 91.3\std{0.35} (\down{0.7}) & \second{92.5}\std{0.32} (\up{0.9}) & 35.8\std{0.91} (\down{1.7}) & 1.82$\times$ & 70 & 4.81 & 523 \\
\hline
\rowcolor{gray!10}
Harmony \venue{MobiSys'23}~\cite{Ouyang2023MobiSys} & 88.2\std{0.71} (\down{3.8}) & 88.8\std{0.64} (\down{2.8}) & \second{42.6}\std{1.12} (\up{5.1}) & 1.00$\times$ & 110 & 4.81 & 839 \\
Pilot \venue{AAAI'25}~\cite{Xiong2025AAAI} & 91.8\std{0.29} (\down{0.2}) & 91.0\std{0.41} (\down{0.6}) & 34.2\std{0.86} (\down{3.3}) & 1.00$\times$ & \second{60} & 4.81 & 855 \\
\rowcolor{gray!10}
FedSA-LoRA \venue{ICLR'25}~\cite{Guo2025ICLR} & 92.0\std{0.33} ($\pm$0.0) & 90.6\std{0.47} (\down{1.0}) & 36.8\std{0.85} (\down{0.7}) & 1.00$\times$ & 70 & 4.81 & 844 \\
HeLoRA \venue{TOIT'25}~\cite{Fan2025TOIT} & 90.4\std{0.48} (\down{1.6}) & 90.7\std{0.39} (\down{0.9}) & 33.5\std{0.98} (\down{4.0}) & 1.45$\times$ & 85 & 4.22 & 614 \\
\rowcolor{gray!10}
FedLEASE \venue{NeurIPS'25}~\cite{Wang2025NeurIPS} & 91.0\std{0.36} (\down{1.0}) & 91.2\std{0.40} (\down{0.4}) & 35.1\std{0.90} (\down{2.4}) & 1.00$\times$ & 75 & 4.81 & 851 \\
\hline
\textbf{\method{} (Ours)} & 90.1\std{0.42} (\down{1.9}) & \best{93.7}\std{0.28} (\up{2.1}) & \best{52.8}\std{1.15} (\up{15.3}) & \second{2.87$\times$} & \best{55} & 4.76 & \second{312} \\
\Xhline{1.2pt}
\end{tabular}}
\end{table*}

\begin{table*}[t]
\centering
\setlength{\tabcolsep}{6pt}
\renewcommand\arraystretch{1.15}
\caption{Main results with Backbone~2 (MOMENT + LoRA/MDLoRA). Best in \best{red}, runner-up in \second{blue}. $^\dagger$Rare-Mod F1: PAMAP2 avg(Mag, HR). Save relative to FedAvg-LoRA.}
\label{tab:main_lora}
\resizebox{0.98\linewidth}{!}{
\begin{tabular}{l|ccc|c|c|c|c|c}
\Xhline{1.2pt}
\rowcolor{CadetBlue!20}
\textbf{Method} & \textbf{PAMAP2 F1 (\%)} & \textbf{MHEALTH F1 (\%)} & \textbf{Rare-Mod F1 (\%)$^\dagger$} & \textbf{Speedup} & \textbf{Comm (KB)} & \textbf{Train. (\%)} & \textbf{Save (\%)} & \textbf{Energy (J)} \\
\Xhline{1.2pt}
FedAvg \venue{AISTATS'17}~\cite{McMahan2017AISTATS} & \best{78.3}\std{0.52} & \second{63.2}\std{0.68} & 4.4\std{0.35} & 1.00$\times$ & 5457 & 0.98 & 0.0 & 1284 \\
\rowcolor{gray!10}
FedProx \venue{MLSys'20}~\cite{Li2020MLSys} & 77.8\std{0.48} (\down{0.5}) & 62.9\std{0.71} (\down{0.3}) & \second{5.1}\std{0.42} (\up{0.7}) & 1.00$\times$ & 5457 & 0.98 & 0.0 & 1291 \\
\hline
FedEL \venue{NeurIPS'25}~\cite{Zhang2025NeurIPS} & 58.4\std{1.35} (\down{19.9}) & 42.1\std{1.52} (\down{21.1}) & 1.2\std{0.28} (\down{3.2}) & \second{7.68$\times$} & \best{2876} & \best{0.48} & \best{47.3} & \second{245} \\
\rowcolor{gray!10}
FedICU \venue{ICML'25}~\cite{Liao2025ICML} & 56.3\std{0.91} (\down{22.0}) & 51.6\std{0.85} (\down{11.6}) & 1.4\std{0.22} (\down{3.0}) & 1.12$\times$ & 4365 & 0.79 & 20.0 & 1158 \\
DarkDistill \venue{KDD'25}~\cite{Qu2025KDD} & 63.8\std{0.78} (\down{14.5}) & 54.5\std{0.82} (\down{8.7}) & 1.8\std{0.31} (\down{2.6}) & 1.68$\times$ & 5457 & 0.72 & 0.0 & 802 \\
\hline
\rowcolor{gray!10}
Harmony \venue{MobiSys'23}~\cite{Ouyang2023MobiSys} & 37.7\std{1.24} (\down{40.6}) & 22.3\std{1.38} (\down{40.9}) & 0.9\std{0.18} (\down{3.5}) & 1.00$\times$ & 5457 & 0.98 & 0.0 & 1276 \\
Pilot \venue{AAAI'25}~\cite{Xiong2025AAAI} & 64.2\std{0.75} (\down{14.1}) & 55.1\std{0.80} (\down{8.1}) & 1.6\std{0.27} (\down{2.8}) & 1.00$\times$ & 5821 & 1.05 & $-$6.7 & 1302 \\
\rowcolor{gray!10}
FedSA-LoRA \venue{ICLR'25}~\cite{Guo2025ICLR} & \second{78.1}\std{0.50} (\down{0.2}) & 53.3\std{0.79} (\down{9.9}) & 4.2\std{0.38} (\down{0.2}) & 1.00$\times$ & 5457 & 0.98 & 0.0 & 1280 \\
HeLoRA \venue{TOIT'25}~\cite{Fan2025TOIT} & 61.4\std{0.82} (\down{16.9}) & 52.8\std{0.84} (\down{10.4}) & 1.5\std{0.26} (\down{2.9}) & 1.35$\times$ & 4092 & 0.71 & 25.0 & 963 \\
\rowcolor{gray!10}
FedLEASE \venue{NeurIPS'25}~\cite{Wang2025NeurIPS} & 65.1\std{0.72} (\down{13.2}) & 56.2\std{0.76} (\down{7.0}) & 1.7\std{0.29} (\down{2.7}) & 1.00$\times$ & 5638 & 1.02 & $-$3.3 & 1295 \\
\hline
\textbf{\method{} (Ours)} & 74.9\std{0.58} (\down{3.4}) & \best{83.4}\std{0.45} (\up{20.2}) & \best{12.3}\std{0.95} (\up{7.9}) & \best{9.41$\times$} & \second{3408} & \second{0.61} & \second{37.5} & \best{178} \\
\Xhline{1.2pt}
\end{tabular}}
\end{table*}

\begin{figure*}[t]
\centering
\subfloat[PAMAP2, B1\label{fig:conv_p_b1}]{%
    \includegraphics[width=0.24\textwidth]{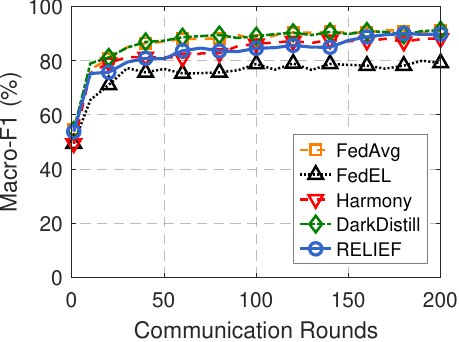}
    }
\hfill
\subfloat[MHEALTH, B1\label{fig:conv_m_b1}]{%
    \includegraphics[width=0.24\textwidth]{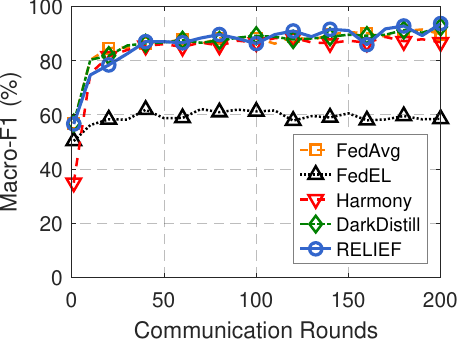}
    }
\hfill
\subfloat[PAMAP2, B2\label{fig:conv_p_b2}]{%
    \includegraphics[width=0.24\textwidth]{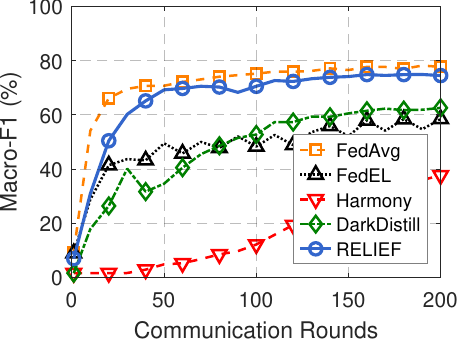}
    }
\hfill
\subfloat[MHEALTH, B2\label{fig:conv_m_b2}]{%
    \includegraphics[width=0.24\textwidth]{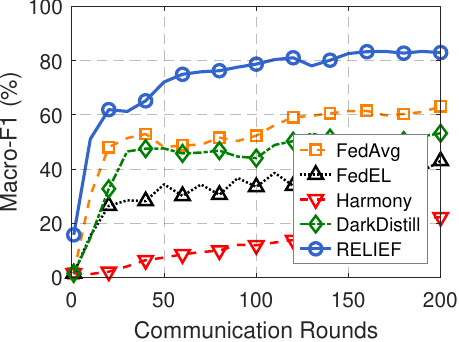}
    }
\caption{Convergence comparison (macro-F1 vs.\ communication round) for five representative methods across two datasets and two backbones.}
\label{fig:convergence}
\end{figure*}

Tables~\ref{tab:main_cnn} and~\ref{tab:main_lora} compare 11 methods across two backbones and two datasets. Under Backbone~1, \method{} attains 93.7\% F1 on MHEALTH (the highest) with 2.87$\times$ speedup and 63\% energy reduction (312\,J vs.\ 847\,J). On PAMAP2 it trades 1.9\,pp for the same acceleration. FedEL~\cite{Zhang2025NeurIPS} is faster (6.83$\times$) but collapses to 81.6\%/62.2\% because its modality-unaware selection assigns weak devices to train absent-sensor parameters. Under Backbone~2, FedEL's advantage reverses (7.68$\times$ vs.\ \method{}'s 9.41$\times$): CNN offers many tensor groups for aggressive pruning, while LoRA's compact structure limits headroom. \method{}'s modality-aligned decomposition maps directly onto the LoRA column blocks and scales with the number of modalities.

Backbone~2 amplifies the contrast. Most baselines within 2\,pp of FedAvg under B1 (e.g., DarkDistill, FedLEASE) drop 13--15\,pp under LoRA, because LoRA's compact space leaves less room to absorb zero-gradient noise. \method{} maintains competitive F1 on PAMAP2 (74.9\% vs.\ 78.3\%) with 9.41$\times$ speedup and 37.5\% communication savings. On MHEALTH it surpasses FedAvg by 20.2\,pp (83.4\% vs.\ 63.2\%), because ECG and IMU occupy different feature spaces, which intensifies cross-modal interference and makes cohort-wise aggregation proportionally more beneficial. The Rare-Mod F1 column confirms that \method{} improves rare modalities by 15.3\,pp (B1) and 7.9\,pp (B2) over FedAvg, with Harmony as the runner-up under B1 (42.6\%) but providing no speedup.

Fig.~\ref{fig:convergence} shows convergence trajectories across all four settings. Under B1, \method{} converges comparably to FedAvg on PAMAP2 and faster on MHEALTH and overtakes all baselines by round 40. FedEL plateaus at a low F1 with persistent oscillation. Under B2, Harmony collapses to 37.7\%/22.3\% because it excludes the LoRA fusion layer from federation, while \method{} is the only method above 70\% on both LoRA settings. The lower absolute F1 under B2 reflects domain mismatch between MOMENT's pretraining corpus and HAR signals, not a limitation of \method{}.

\subsection{Ablation Study}

\begin{table*}[t]
\centering
\setlength{\tabcolsep}{4pt}
\renewcommand\arraystretch{1.15}
\caption{Ablation study. Arrows show delta relative to V0 (full \method{}).}
\label{tab:ablation}
\resizebox{0.95\linewidth}{!}{
\begin{tabular}{cl|ccc|ccc}
\Xhline{1.2pt}
\rowcolor{CadetBlue!20}
& & \multicolumn{3}{c|}{\textbf{Backbone~1 (CNN)}} & \multicolumn{3}{c}{\textbf{Backbone~2 (LoRA)}} \\
\rowcolor{CadetBlue!20}
& \textbf{Variant} & \textbf{PAMAP2 F1 (\%)} & \textbf{MHEALTH F1 (\%)} & \textbf{Speedup} & \textbf{PAMAP2 F1 (\%)} & \textbf{MHEALTH F1 (\%)} & \textbf{Speedup} \\
\Xhline{1.2pt}
V0 & \method{} (full) & 90.1\std{0.42} & 93.7\std{0.28} & 2.87$\times$ & 74.9\std{0.58} & 83.4\std{0.45} & 9.41$\times$ \\
\rowcolor{gray!10}
V1 & w/o elastic training & 94.0\std{0.25} (\up{3.9}) & 95.5\std{0.19} (\up{1.8}) & 1.66$\times$ & 78.5\std{0.45} (\up{3.6}) & 85.7\std{0.33} (\up{2.3}) & 1.52$\times$ \\
V2 & w/o cohort-wise agg. & 83.7\std{0.89} (\down{6.4}) & 86.5\std{0.78} (\down{7.2}) & 3.84$\times$ & 66.8\std{1.15} (\down{8.1}) & 73.8\std{1.05} (\down{9.6}) & 12.6$\times$ \\
\rowcolor{gray!10}
V3 & random elastic alloc. & 83.1\std{0.95} (\down{7.0}) & 85.5\std{0.88} (\down{8.2}) & 3.84$\times$ & 65.3\std{1.22} (\down{9.6}) & 72.1\std{1.19} (\down{11.3}) & 12.6$\times$ \\
\Xhline{1.2pt}
\end{tabular}}
\end{table*}

Table~\ref{tab:ablation} isolates each component's contribution. Removing cohort-wise aggregation (V2) drops F1 by 6.4--7.2\,pp (B1) and 8.1--9.6\,pp (B2), with the larger drop on MHEALTH ($-$7.2/$-$9.6\,pp) reflecting its smaller ECG cohort. Random allocation (V3) causes the largest drop ($-$7.0 to $-$8.2\,pp on B1, $-$9.6 to $-$11.3\,pp on B2). V2 and V3 share the same budget $k_n$ and speedup (3.84$\times$/12.6$\times$), which exceeds V0 because V0's mandatory inclusion constraint raises the minimum workload.

Disabling elastic training (V1) raises F1 by 1.8--3.9\,pp but cuts speedup to 1.66$\times$/1.52$\times$ and increases energy from 312\,J to 578\,J (B1). This trade-off is by design: in latency-sensitive IoT deployments, V0 is justified by 2.87$\times$ faster rounds; in accuracy-critical scenarios, V1 remains a strong standalone option.

\subsection{Sensitivity Analysis}

\begin{table*}[t]
\centering
\setlength{\tabcolsep}{4pt}
\renewcommand\arraystretch{1.15}
\caption{Sensitivity analysis on PAMAP2. F1 (\%) under varying heterogeneity and client count.}
\label{tab:sens_pamap2}
\resizebox{0.95\linewidth}{!}{
\begin{tabular}{ll|ccccc|ccccc}
\Xhline{1.2pt}
\rowcolor{CadetBlue!20}
& & \multicolumn{5}{c|}{\textbf{Backbone~1 (CNN)}} & \multicolumn{5}{c}{\textbf{Backbone~2 (LoRA)}} \\
\rowcolor{CadetBlue!20}
\textbf{Factor} & \textbf{Setting} & \textbf{FedAvg} & \textbf{FedEL} & \textbf{Harmony} & \textbf{DarkDist.} & \textbf{\method{}} & \textbf{FedAvg} & \textbf{FedEL} & \textbf{Harmony} & \textbf{DarkDist.} & \textbf{\method{}} \\
\Xhline{1.2pt}
Hetero. & Mild (10$\times$) & 92.1\std{0.37} & 82.4\std{1.02} & 88.5\std{0.68} & 91.5\std{0.41} & \cellcolor{green!8}90.4\std{0.39} & 78.8\std{0.48} & 60.2\std{1.22} & 38.9\std{1.15} & 64.5\std{0.79} & \cellcolor{green!8}75.6\std{0.54} \\
\rowcolor{gray!10}
& Moderate (55$\times$) & 92.0\std{0.38} & 81.6\std{1.05} & 88.2\std{0.71} & 91.3\std{0.35} & \cellcolor{green!8}90.1\std{0.42} & 78.3\std{0.52} & 58.4\std{1.35} & 37.7\std{1.24} & 63.8\std{0.78} & \cellcolor{green!8}74.9\std{0.58} \\
& Extreme (100$\times$) & 91.6\std{0.43} & 80.9\std{1.18} & 87.8\std{0.82} & 90.7\std{0.39} & \cellcolor{green!8}88.5\std{0.51} & 77.9\std{0.59} & 57.8\std{1.42} & 36.2\std{1.35} & 62.4\std{0.91} & \cellcolor{green!8}73.6\std{0.62} \\
\hline
\rowcolor{gray!10}
Scale & $N{=}8$ & 92.0\std{0.38} & 81.6\std{1.05} & 88.2\std{0.71} & 91.3\std{0.35} & \cellcolor{green!8}90.1\std{0.42} & 78.3\std{0.52} & 58.4\std{1.35} & 37.7\std{1.24} & 63.8\std{0.78} & \cellcolor{green!8}74.9\std{0.58} \\
& $N{=}20$ & 91.4\std{0.40} & 80.2\std{1.15} & 87.5\std{0.79} & 90.6\std{0.38} & \cellcolor{green!8}89.6\std{0.45} & 77.6\std{0.55} & 57.1\std{1.38} & 35.8\std{1.28} & 62.1\std{0.85} & \cellcolor{green!8}74.2\std{0.61} \\
\rowcolor{gray!10}
& $N{=}50$ & 90.5\std{0.49} & 78.8\std{1.28} & 86.1\std{0.88} & 89.4\std{0.46} & \cellcolor{green!8}88.7\std{0.53} & 76.8\std{0.64} & 55.6\std{1.52} & 33.5\std{1.41} & 60.2\std{0.98} & \cellcolor{green!8}73.4\std{0.66} \\
& $N{=}100$ & 89.8\std{0.54} & 77.1\std{1.35} & 84.8\std{0.95} & 88.5\std{0.52} & \cellcolor{green!8}88.2\std{0.57} & 76.1\std{0.70} & 54.2\std{1.58} & 31.8\std{1.48} & 58.5\std{1.05} & \cellcolor{green!8}72.8\std{0.72} \\
\Xhline{1.2pt}
\end{tabular}}
\end{table*}

\begin{table*}[t]
\centering
\setlength{\tabcolsep}{4pt}
\renewcommand\arraystretch{1.15}
\caption{Sensitivity analysis on MHEALTH. F1 (\%) under varying heterogeneity and client count.}
\label{tab:sens_mhealth}
\resizebox{0.95\linewidth}{!}{
\begin{tabular}{ll|ccccc|ccccc}
\Xhline{1.2pt}
\rowcolor{CadetBlue!20}
& & \multicolumn{5}{c|}{\textbf{Backbone~1 (CNN)}} & \multicolumn{5}{c}{\textbf{Backbone~2 (LoRA)}} \\
\rowcolor{CadetBlue!20}
\textbf{Factor} & \textbf{Setting} & \textbf{FedAvg} & \textbf{FedEL} & \textbf{Harmony} & \textbf{DarkDist.} & \textbf{\method{}} & \textbf{FedAvg} & \textbf{FedEL} & \textbf{Harmony} & \textbf{DarkDist.} & \textbf{\method{}} \\
\Xhline{1.2pt}
Hetero. & Mild (10$\times$) & 91.8\std{0.42} & 63.5\std{1.32} & 89.1\std{0.65} & 92.8\std{0.33} & \cellcolor{green!8}94.0\std{0.31} & 63.8\std{0.62} & 43.5\std{1.48} & 23.8\std{1.28} & 55.7\std{0.85} & \cellcolor{green!8}83.8\std{0.44} \\
\rowcolor{gray!10}
& Moderate (55$\times$) & 91.6\std{0.45} & 62.2\std{1.42} & 88.8\std{0.64} & 92.5\std{0.32} & \cellcolor{green!8}93.7\std{0.28} & 63.2\std{0.68} & 42.1\std{1.52} & 22.3\std{1.38} & 54.5\std{0.82} & \cellcolor{green!8}83.4\std{0.45} \\
& Extreme (100$\times$) & 91.3\std{0.47} & 61.5\std{1.55} & 88.1\std{0.78} & 91.6\std{0.38} & \cellcolor{green!8}92.1\std{0.41} & 62.5\std{0.75} & 41.2\std{1.58} & 21.5\std{1.45} & 53.2\std{0.93} & \cellcolor{green!8}81.5\std{0.54} \\
\hline
\rowcolor{gray!10}
Scale & $N{=}10$ & 91.6\std{0.45} & 62.2\std{1.42} & 88.8\std{0.64} & 92.5\std{0.32} & \cellcolor{green!8}93.7\std{0.28} & 63.2\std{0.68} & 42.1\std{1.52} & 22.3\std{1.38} & 54.5\std{0.82} & \cellcolor{green!8}83.4\std{0.45} \\
& $N{=}20$ & 91.0\std{0.48} & 60.8\std{1.48} & 87.9\std{0.75} & 91.8\std{0.36} & \cellcolor{green!8}93.1\std{0.34} & 62.4\std{0.71} & 40.5\std{1.62} & 20.8\std{1.42} & 52.8\std{0.89} & \cellcolor{green!8}82.1\std{0.52} \\
\rowcolor{gray!10}
& $N{=}50$ & 90.1\std{0.53} & 57.3\std{1.65} & 86.5\std{0.85} & 90.4\std{0.44} & \cellcolor{green!8}92.0\std{0.38} & 61.2\std{0.78} & 38.2\std{1.68} & 18.9\std{1.55} & 50.6\std{0.97} & \cellcolor{green!8}80.5\std{0.61} \\
& $N{=}100$ & 89.4\std{0.60} & 54.6\std{1.75} & 85.2\std{0.92} & 89.1\std{0.51} & \cellcolor{green!8}91.2\std{0.48} & 60.3\std{0.85} & 36.1\std{1.78} & 17.2\std{1.62} & 48.3\std{1.05} & \cellcolor{green!8}79.2\std{0.68} \\
\Xhline{1.2pt}
\end{tabular}}
\end{table*}

Tables~\ref{tab:sens_pamap2} and~\ref{tab:sens_mhealth} examine robustness under varying heterogeneity and fleet size. As the compute gap widens from 10$\times$ to 100$\times$, FedAvg F1 remains nearly constant, while FedEL drops steadily (82.4$\to$80.9\% on PAMAP2 B1). \method{} degrades by only 1.9\,pp on both datasets under B1 and maintains the best accuracy-robustness trade-off.

Increasing $N$ from 8 to 100 reduces all methods' F1 as data becomes more fragmented. FedEL degrades most sharply ($-$4.5/$-$7.6\,pp under B1) because modality-unaware selection makes more errors with more heterogeneous devices. Harmony under B2 collapses from 22.3\% to 17.2\% on MHEALTH. \method{} drops only 1.9/2.5\,pp under B1, and its advantage over FedAvg on MHEALTH persists at all fleet sizes (83.4$\to$79.2\% vs.\ 63.2$\to$60.3\% under B2).

\subsection{In-Depth Analysis}

\begin{figure*}[t]
\centering
\subfloat[PAMAP2, B1\label{fig:pmod_p_b1}]{%
    \includegraphics[width=0.24\textwidth]{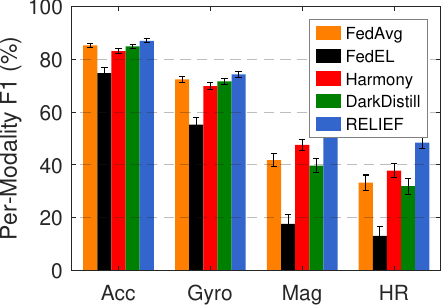}
    }
\hfill
\subfloat[MHEALTH, B1\label{fig:pmod_m_b1}]{%
    \includegraphics[width=0.24\textwidth]{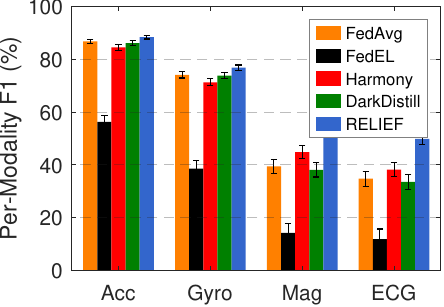}
    }
\hfill
\subfloat[PAMAP2, B2\label{fig:pmod_p_b2}]{%
    \includegraphics[width=0.24\textwidth]{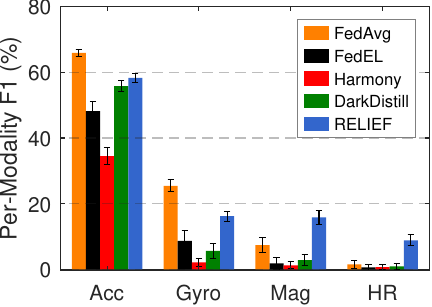}
    }
\hfill
\subfloat[MHEALTH, B2\label{fig:pmod_m_b2}]{%
    \includegraphics[width=0.24\textwidth]{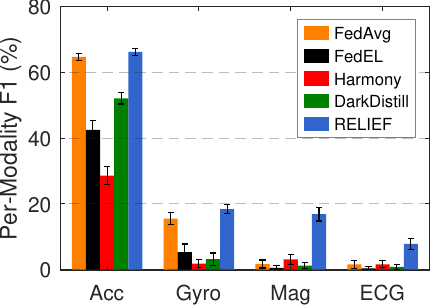}
    }
\caption{Per-modality F1 breakdown. \method{} yields the largest gains on rare modalities (Mag, HR, ECG), consistent with the theoretical prediction that cohort-wise aggregation benefits smaller cohorts more.}
\label{fig:permod_f1}
\end{figure*}

Fig.~\ref{fig:permod_f1} breaks down F1 by modality. Acc F1 varies by less than 3\,pp across methods (all devices contribute Acc gradients), which confirms that cohort-wise aggregation does not harm shared modalities. In contrast, under B1 on PAMAP2, \method{} improves Mag F1 by 15.4\,pp and HR F1 by 15.2\,pp over FedAvg, while Acc improves by only 1.7\,pp. Under B2 on MHEALTH, Mag F1 jumps from 1.7\% to 16.9\% and ECG from 1.6\% to 7.8\%. Harmony improves rare modalities under B1 but collapses under B2 because it excludes the fusion layer. The disproportionate rare-modality gain is consistent with Theorem~\ref{thm:convergence}: the cohort residual $(\sigma^2/E + \zeta_m^2)/|\tilde{\mathcal{C}}_m|$ decreases faster for small cohorts.

\section{Real-Device Deployment}
\label{sec:deployment}

The simulation experiments in Section~\ref{sec:experiments} rely on FLOP-proportional timing models and datasheet-calibrated energy estimates. To validate that \method{}'s efficiency gains transfer to physical hardware, we deploy the framework on a testbed of NVIDIA Jetson AGX Orin devices and measure training time, power draw, and communication volume under real edge computing constraints. Two complementary experiment groups target different granularities: per-device profiling (2 clients) isolates the time and energy savings at the individual device level, while system-level validation (8/10 clients) verifies end-to-end FL performance under the full heterogeneous fleet configuration.

\subsection{Testbed Configuration}
\label{sec:testbed}

\begin{figure}[t]
\centering
\includegraphics[width=0.85\columnwidth]{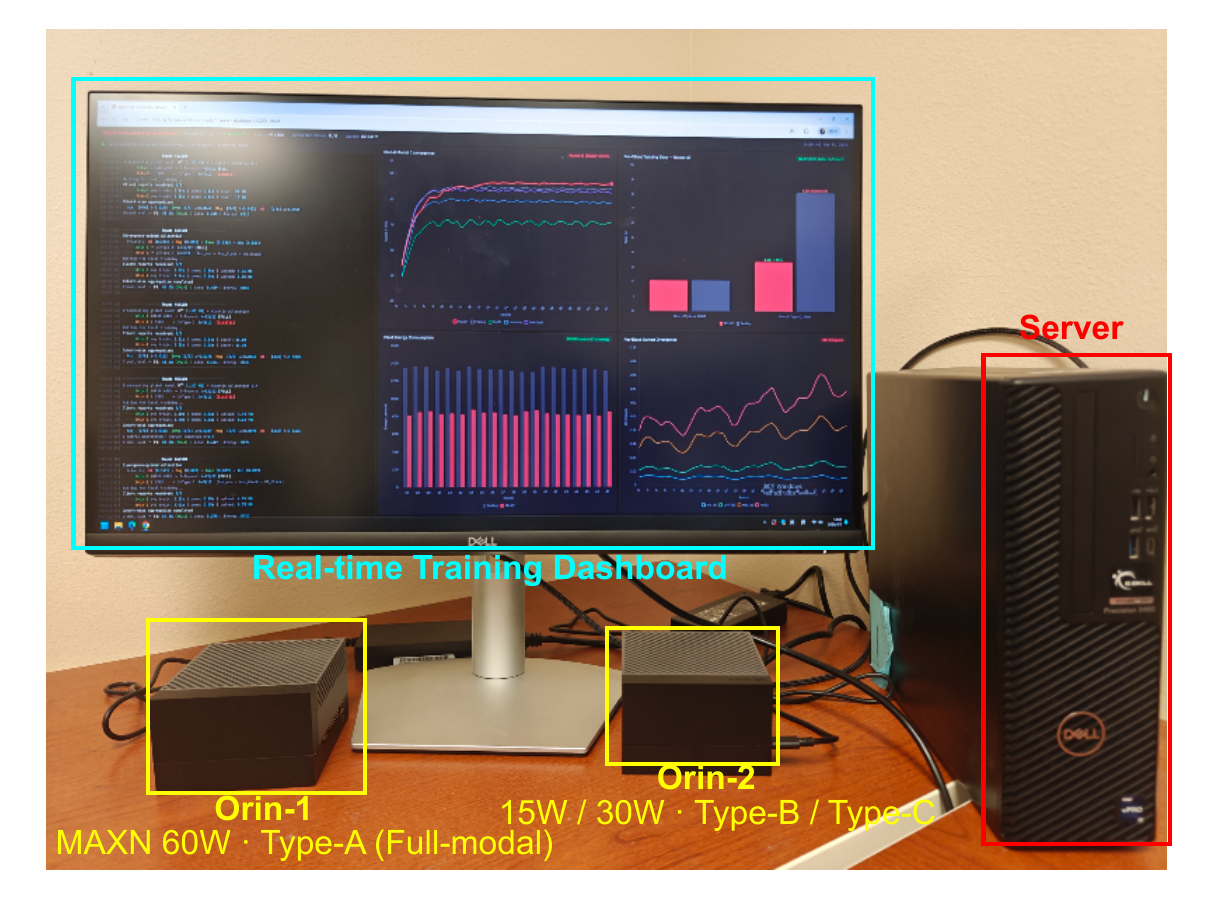}
\caption{Real-device testbed deployment. The server coordinates federated training across two Jetson AGX Orin 64\,GB devices operating at different power modes to emulate heterogeneous IoT edge devices.}
\label{fig:testbed}
\end{figure}

\begin{table}[t]
\centering
\setlength{\tabcolsep}{4pt}
\renewcommand\arraystretch{1.15}
\caption{Real-device testbed configuration.}
\label{tab:testbed}
\begin{tabular}{llll}
\Xhline{1.2pt}
\rowcolor{CadetBlue!20}
\textbf{Device} & \textbf{Power Mode} & \textbf{Modality} & \textbf{Device Type} \\
\Xhline{1.2pt}
Orin-1 & MAXN (60\,W) & All 4 & Type-A (Full) \\
\rowcolor{gray!10}
Orin-2 & 30\,W & Acc+Gyro & Type-B (Mid) \\
Orin-2 & 15\,W & Acc-only & Type-C (Low) \\
\Xhline{1.2pt}
\end{tabular}
\end{table}

Table~\ref{tab:testbed} summarizes the testbed configuration and Fig.~\ref{fig:testbed} shows the physical deployment. The aggregation server runs on an Intel Core i7 workstation with 16\,GB RAM and an NVIDIA RTX A2000 GPU under Ubuntu 22.04. It runs the Flower federated learning framework~\cite{Beutel2020ARXIV}
for server-side aggregation and divergence computation. Two NVIDIA Jetson AGX Orin 64\,GB modules serve as edge clients, connected to the server through a dedicated 802.11ac WiFi access point. Communication between the server and clients uses gRPC over TCP.

To emulate the coupled system-modality heterogeneity from Section~\ref{sec:problem_formulation} with two physical devices, we exploit the Jetson's configurable power modes via \texttt{nvpmodel}. Orin-1 operates at MAXN mode (60\,W) as a Type-A (full-modality, high-compute) client, while Orin-2 operates at 15\,W or 30\,W to emulate Type-C (single-modality, low-compute) or Type-B (dual-modality, mid-compute) clients (Table~\ref{tab:testbed}). Both devices share 64\,GB of unified CPU-GPU memory, which ensures that both Backbone~1 (CNN) and Backbone~2 (MOMENT + LoRA) execute without memory constraints at all power levels. We compare the same five representative methods used in the simulation sensitivity analysis: FedAvg~\cite{McMahan2017AISTATS}, FedEL~\cite{Zhang2025NeurIPS}, DarkDistill~\cite{Qu2025KDD}, Harmony~\cite{Ouyang2023MobiSys}, and \method{}. All hyperparameters (learning rate, batch size, local epochs, EMA coefficient) match the simulation settings in Section~\ref{sec:experiments}. Experiments cover both datasets (PAMAP2, MHEALTH) and both backbones.

\subsection{Per-Device Profiling}
\label{sec:profiling}

The first experiment group assigns each Jetson to exactly one FL client: Orin-1 (60\,W, Type-A) and Orin-2 (15\,W, Type-C). Both devices train in parallel with exclusive GPU access, which provides interference-free measurements at the individual device level.

\paragraph{Measurement protocol.}
Each round proceeds as follows: the server broadcasts the global model and the elastic allocation $\mathcal{S}_n$ to both devices; both devices perform local training in parallel; the faster device waits until the slower one completes; and model updates are uploaded via gRPC. We record six per-device metrics every round: (1)~training time via \texttt{torch.cuda.Event} timestamps, (2)~communication time via gRPC transfer timing, (3)~idle waiting time as the residual of the round duration, (4)~instantaneous power at 100\,ms intervals from the on-board INA3221 power monitor accessed through \texttt{sysfs}, (5)~per-round energy as the time integral $\int P(t)\,dt$, and (6)~upload payload size from the actual gRPC message. Each configuration runs for 100 rounds with 3 random seeds across both datasets and both backbones. Per-device timing and energy values reported below are averaged over all rounds and seeds.

\begin{figure*}[t]
\centering
\subfloat[Time, PAMAP2 B1\label{fig:tb_p_b1}]{%
    \includegraphics[width=0.32\textwidth]{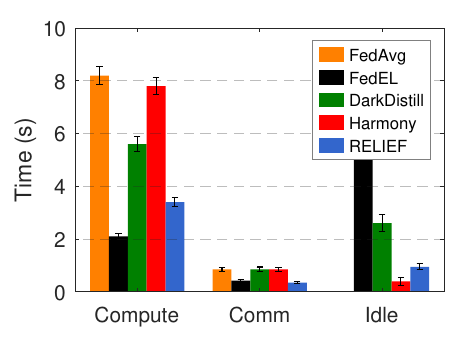}
    }
\hfill
\subfloat[Power, PAMAP2 B1\label{fig:pw_p_b1}]{%
    \includegraphics[width=0.32\textwidth]{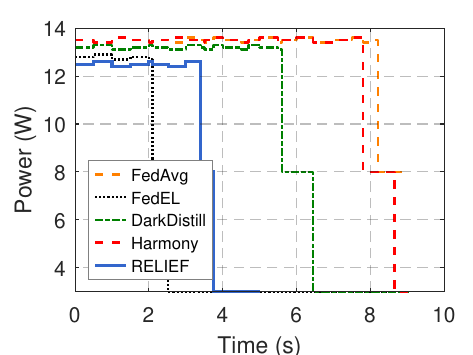}
    }
\hfill
\subfloat[Energy, PAMAP2 B1\label{fig:ef_p_b1}]{%
    \includegraphics[width=0.32\textwidth]{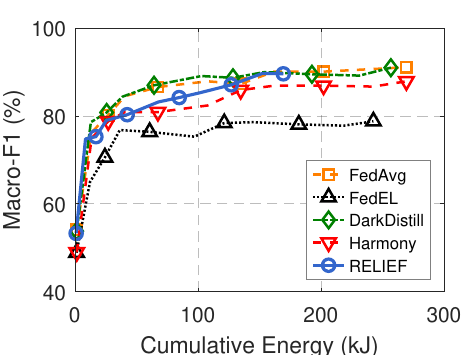}
    }\\[4pt]
\subfloat[Time, MHEALTH B2\label{fig:tb_m_b2}]{%
    \includegraphics[width=0.32\textwidth]{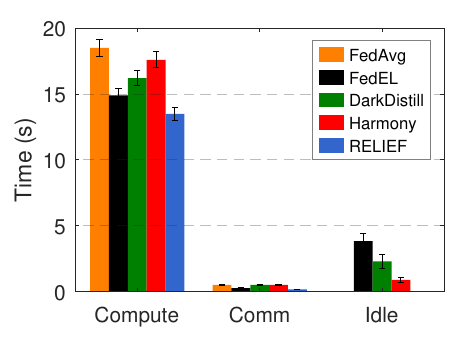}
    }
\hfill
\subfloat[Power, MHEALTH B2\label{fig:pw_m_b2}]{%
    \includegraphics[width=0.32\textwidth]{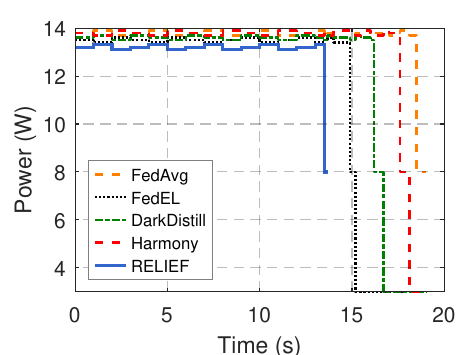}
    }
\hfill
\subfloat[Energy, MHEALTH B2\label{fig:ef_m_b2}]{%
    \includegraphics[width=0.32\textwidth]{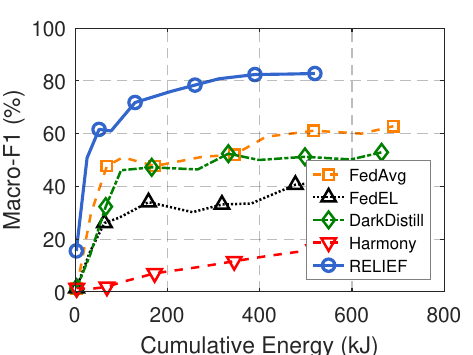}
    }
\caption{Real-device deployment results. \textbf{Left column:} per-round time breakdown of the Type-C device (compute, communication, idle). \textbf{Middle column:} instantaneous power trace during one training round. \textbf{Right column:} macro-F1 vs.\ cumulative fleet energy. Top row: PAMAP2 with Backbone~1 (CNN); bottom row: MHEALTH with Backbone~2 (LoRA).}
\label{fig:deployment}
\end{figure*}

\subsection{Results and Analysis}
\label{sec:deploy_results}

Fig.~\ref{fig:deployment} presents all profiling and energy results. Under Backbone~1 (Fig.~\ref{fig:tb_p_b1}), FedAvg's Type-C consumes 8.2\,s of compute per round, which locks the round time at 9.05\,s for all non-\method{} methods. FedEL finishes in 2.1\,s but idles for 6.53\,s. \method{} reduces compute to 3.4\,s and shifts the bottleneck to Type-A, so the round time drops to 4.70\,s (1.93$\times$ speedup). Under Backbone~2 (Fig.~\ref{fig:tb_m_b2}), frozen MOMENT encoders impose a fixed ${\sim}$12\,s forward-pass cost; only the backward pass is reducible (FedAvg ${\sim}$6\,s vs.\ \method{} ${\sim}$1.5\,s). The round time decreases from 19.02\,s to 13.68\,s (1.39$\times$), with the smaller speedup reflecting the dominance of the fixed forward cost.

The power traces (Fig.~\ref{fig:pw_p_b1},~\ref{fig:pw_m_b2}) follow a three-phase pattern: active compute (${\sim}$13\,W), communication (${\sim}$8\,W), and idle (${\sim}$3\,W). The area under each curve gives the per-round device energy: 48\,J (\method{}) vs.\ 118\,J (FedAvg) under B1 ($-$59\%), and 180\,J vs.\ 259\,J under B2 ($-$30\%). We profile each device type separately and compute fleet-level energy as $3E_A + 3E_B + 2E_C$ for PAMAP2 and $3E_A + 3E_B + 4E_C$ for MHEALTH, with idle power included (Type-A ${\sim}$12\,W, Type-B ${\sim}$6\,W, Type-C ${\sim}$3\,W). Under B1, per-round fleet energy is 846\,J (\method{}) vs.\ 1346\,J (FedAvg), a 37\% saving.

The energy-efficiency curves (Fig.~\ref{fig:ef_p_b1},~\ref{fig:ef_m_b2}) combine the F1 trajectory from the full 200-round runs with per-round fleet energy from profiling. Under B1, \method{} reaches 89.6\% F1 at 169\,kJ; FedAvg requires 269\,kJ for 91.1\% F1. FedEL uses 242\,kJ but plateaus at 78.8\% ($-$11\,pp). Under B2, \method{} reaches 82.8\% at 520\,kJ; FedAvg reaches 63.1\% at 690\,kJ. Harmony stays at 22\% regardless of budget.

\method{}'s real-device F1 differs from simulation by less than 1\,pp; most baselines fall within 2\,pp, with FedEL showing the largest gap (2.8\,pp). Under B1, real speedup (1.93$\times$) is close to the simulated 2.87$\times$ because forward and backward passes contribute roughly equally. Under B2, real speedup (1.39$\times$) is lower than the simulated 9.41$\times$ because the FLOP-proportional model does not separate fixed forward cost from reducible backward cost. Subtracting the shared ${\sim}$12\,s forward time yields a backward-only speedup of $6.5/1.5 = 4.3\times$, consistent with \method{} training ${\sim}$23\% of the LoRA parameters; the remaining gap is attributable to the smaller hardware compute ratio (4$\times$ vs.\ 55$\times$ in simulation). This confirms that the simulation correctly predicts training-compute reduction, while the overall speedup gap is a systematic artifact of FLOP-proportional timing applied to pretrained architectures with frozen forward passes.

\section{Conclusion}\label{sec:con}

In this paper, we have investigated federated learning over heterogeneous IoT edge networks where system, modality, and data heterogeneity are coupled through the device cost gradient. We have identified that cross-modal gradient interference propagates beyond missing-modality blocks to corrupt shared-modality representations, and that rare-modality update divergence amplifies rather than converges under standard aggregation. To address these challenges, we have proposed \method{}, a unified framework that leverages the modality-aligned column-block structure of the fusion-layer LoRA matrix as a shared interface for cohort-wise aggregation, divergence-guided elastic training, and on-demand communication. Our theoretical analysis shows that cohort-wise aggregation eliminates cross-modal interference from the convergence bound and that divergence-guided allocation achieves sublinear regret. Our evaluation on two IoT sensor datasets under both full-parameter (CNN) and parameter-efficient (LoRA) training demonstrates that \method{} achieves up to 9.41$\times$ wall-clock speedup and 37\% energy reduction over FedAvg while improving rare-modality F1 by up to 15.3\,pp. Real-device deployment on a two-Jetson AGX Orin testbed confirms these gains on physical hardware.

\bibliographystyle{IEEEtran}
\bibliography{egbib}

\begin{IEEEbiographynophoto}{Beining Wu} (Member, IEEE) received the B.S. degree in mathematics and applied mathematics from Anhui Normal University, Wuhu, China, in 2024. He is currently working toward the Ph.D. degree in computer science with South Dakota State University (SDSU), Brookings, SD, USA. His research interests include continual learning, and Multimodal LLM.
\end{IEEEbiographynophoto}

\begin{IEEEbiographynophoto}{Zihao Ding} received the B.S. degree in electronic and information engineering from Anhui Normal University, Wuhu, China, in 2025. He is currently working toward the Ph.D. degree in computer science with South Dakota State University (SDSU), Brookings, SD, USA. He received the Best Paper Award at IEEE IPCCC. His research focuses on machine unlearning.
\end{IEEEbiographynophoto}

\begin{IEEEbiographynophoto}{Jun Huang} (Senior Member, IEEE) received the Ph.D. degree from Beijing University of Posts and Telecommunications, Beijing, China, in 2012. He is currently an Assistant Professor with the Department of Electrical Engineering and Computer Science, South Dakota State University, Brookings, SD, USA. He  was a Guest Professor at the National Institute of Standards and Technology. His honors include the Best Paper Award from IEEE IPCCC (2025), Outstanding Research Award (Tier I) from CQUPT (2019), Best Paper Award from EAI Mobimedia (2019), Outstanding Service Awards from ACM RACS (2017--2019), Best Paper Nomination from ACM SAC (2014), and Best Paper Award from AsiaFI (2011). He currently serves as an Associate Editor for \textsc{IEEE Internet of Things Journal}, Elsevier Digital Communications and Networks, ICT Express, and IET Wireless Sensor Systems, and as Technical Editor for ACM SIGAPP Applied Computing Review. He has served as chair or co-chair for multiple conferences and workshops at major IEEE and ACM events.
\end{IEEEbiographynophoto}

\end{document}